\documentclass{lmcs} 
\pdfoutput=1
\usepackage[utf8]{inputenc}

\usepackage{lastpage}
\lmcsdoi{20}{4}{27}
\lmcsheading{}{\pageref{LastPage}}{}{}%
{May~24,~2023}{Dec.~24,~2024}{}


\usepackage{hyperref}
\theoremstyle{plain} 

\usepackage{enumitem}
\usepackage{amsmath}
\usepackage{amsthm}
\usepackage{amssymb}
\usepackage{algorithmic}
\usepackage{graphicx}
\usepackage{xcolor}
\usepackage{url}
\usepackage{tikz}
\usepackage{mathrsfs}
\usepackage{comment}
\usetikzlibrary{backgrounds, arrows.meta, positioning}

\newcommand{\bi}{\begin{itemize}}
\newcommand{\ei}{\end  {itemize}}
\newcommand{\bt}{\begin{tabbing}}
\newcommand{\et}{\end  {tabbing}}
\newcommand{\be}{\begin{enumerate}}
\newcommand{\ee}{\end  {enumerate}}

\newcommand{\pebbleleft}{\mathsf{Pebble\text{-}Left}}
\newcommand{\pebbleright}{\mathsf{Pebble\text{-}Right}}
\newcommand{\splitleft}{\mathsf{Split\text{-}Left}}
\newcommand{\splitright}{\mathsf{Split\text{-}Right}}
\newcommand{\swap}{\mathsf{Swap}}
\newcommand{\close}{\mathsf{Close}}

\newcommand{\playleft}{\mathsf{Play\text{-}Left}}
\newcommand{\playright}{\mathsf{Play\text{-}Right}}

\newcommand{\bA}{{\bf A}}
\newcommand{\bB}{{\bf B}}
\newcommand{\bC}{{\bf C}}
\newcommand{\bS}{{\bf S}}

\newcommand{\commentout}[1]{}

\newcommand{\ef}{EF }
\newcommand{\EF}{Ehrenfeucht-Fra\"{i}ss\'{e} }
\newcommand{\ms}{MS }

\newcommand{\MS}{multi-structural }
\newcommand{\qvt}{QVT }

\newcommand{\fo}{{\rm FO}}

\newcommand{\np}{\mbox{{\rm NP}}}

\newcommand{\p}{\mbox{{\rm P}}}

\renewcommand{\phi}{\varphi}       

\newcommand{\cA}{\mathcal{A}}
\newcommand{\cB}{\mathcal{B}}
\newcommand{\cC}{\mathcal{C}}

\newcommand{\cS}{\mathcal{S}}
\newcommand{\cT}{\mathcal{T}}
\newcommand{\cL}{\mathcal{L}}

\newcommand{\sA}{\mathscr{A}}
\newcommand{\sB}{\mathscr{B}}

\newcommand{\N}{\mathbb{N}}

\newcommand{\QN}{\mathrm{QN}}
\newcommand{\DSPACE}{\mathrm{DSPACE}}
\newcommand{\NL}{\mathrm{NL}}

\newcommand{\lo}{\mathsf{LO}}
\newcommand{\rt}{\mathsf{RT}}
\newcommand{\Xr}{X_\textrm{root}}

\newcommand{\pred}{\mathsf{Pred}}

\newcommand{\gnode}[3]{\langle ({#1},{#2} )\,\big\vert\, #3 \rangle}
\newcommand{\QVT}[4]{\mathcal{QVT}({#2},{#1})\text{ game on }(#3, #4)}
\newcommand{\SG}[3]{\mathcal{SG}(#1, #2; #3)}

\theoremstyle{plain}
\newtheorem{method}[thm]{Method}
\newtheorem{game}[thm]{Game}

\begin{document}

\title[Multi-Structural Games and Beyond]{Multi-Structural Games and Beyond}

\author[M.~Carmosino]{Marco Carmosino\lmcsorcid{0009-0007-1118-1352}}[a]
\author[R.~Fagin]{Ronald Fagin\lmcsorcid{0000-0002-7374-0347}}[a]
\author[N.~Immerman]{Neil Immerman\lmcsorcid{0000-0001-6609-5952}}[b]
\author[P.~Kolaitis]{Phokion G.~Kolaitis\lmcsorcid{0000-0002-8407-8563}}[a,c]
\author[J.~Lenchner]{Jonathan Lenchner\lmcsorcid{0000-0002-9427-8470}}[a]
\author[R.~Sengupta]{Rik Sengupta\lmcsorcid{0000-0002-9238-5408}}[a,b]

\address{IBM Research}	
\email{mlc@ibm.com, fagin@us.ibm.com, lenchner@us.ibm.com, rik@ibm.com}  

\address{University of Massachusetts Amherst}	
\email{immerman@umass.edu}  

\address{University of California Santa Cruz}	
\email{kolaitis@ucsc.edu}  





\begin{abstract}
  \noindent Multi-structural (MS) games are combinatorial games that capture the number of quantifiers of first-order sentences. On the face of their definition, MS games differ from Ehrenfeucht-Fra\"{i}ss\'{e} (EF) games in two ways: first, MS games are played on two sets of structures, while EF games are played on a pair of structures; second, in MS games, Duplicator can make any number of copies of structures. In the first part of this paper, we perform a finer analysis of MS games and develop a closer comparison of MS games with EF games. In particular, we point out that the use of sets of structures is of the essence and that when MS games are played on pairs of structures,  they capture Boolean combinations of first-order sentences with a fixed number of quantifiers. After this, we focus on another important difference between MS games and EF games, namely, the necessity for Spoiler to play on top of a previous move in order to win some MS games. Via an analysis of the types realized during MS games, we delineate the expressive power of the variant of MS games in which Spoiler never plays on top of a previous move. In the second part we focus on simultaneously capturing number of quantifiers and number of variables in first-order logic. We show that natural variants of the MS game do \textit{not} achieve this. We then introduce a new game, the quantifier-variable tree game, and show that it simultaneously captures the number of quantifiers and number of variables. We conclude by generalizing this game to a family of games, the \emph{syntactic games}, that simultaneously capture reasonable syntactic measures and the number of variables.
\end{abstract}

\maketitle


\section{Introduction}\label{sec:intro}

Combinatorial games are an effective tool to analyze the expressive power of logics on sets of structures. The prototypical example of such combinatorial games are the \EF (EF) games \cite{Ehr61,Fra54}, played by two players, called Spoiler and Duplicator; \ef games capture the quantifier rank of first-order sentences and have been used to analyze the expressive power of first-order logic ({\fo}). Specifically, for every $r \in \N$, two structures $\bA$ and $\bB$ satisfy the same {\fo}-sentences of quantifier rank $r$ if and only if Duplicator wins the $r$-round \ef game on $(\bA,\bB)$.
Since this holds true for finite and infinite structures alike, \ef games can be used in finite model theory, unlike other tools from mathematical logic that fail in the finite realm, such as the compactness theorem. Variants of \ef games in which the players start by playing unary relations were used to analyze the expressive power of monadic \np, which is the collection of  problems in \np\ that are definable by sentences of existential monadic second-order logic \cite{DBLP:journals/mlq/Fagin75,DBLP:conf/stoc/Rougemont84,DBLP:conf/focs/Schwentick94,DBLP:journals/iandc/FaginSV95}. A
different variant of \ef games are the pebble games; they capture the number of variables and were used to analyze the expressive power of finite-variable first-order logic and infinitary logics
\cite{DBLP:journals/jsyml/Barwise77,DBLP:journals/jcss/Immerman82,DBLP:journals/iandc/KolaitisV92}.

Multi-structural (MS) games were introduced in \cite{Immerman81},  re-discovered recently in \cite{MSgames1}, and investigated further in \cite{MSgames2} and \cite{MSgames3}. \ms games capture the number of quantifiers of {\fo}-sentences, which is another important parameter in understanding the expressive power of {\fo}. If one is interested in showing that a property $P$ of finite structures, such as connectivity or acyclicity, is not expressible by \emph{any} fixed {\fo}-sentence $\psi$, then both \ef games and \ms games work equally well. The original motivation of introducing \ms games, however, came from the potential use of \ms games to obtain lower bounds when a sequence $\{\psi_n\}_{n\geq 1}$ of {\fo}-sentences is considered so that a property $P$ is expressed by $\psi_n$ on structures with at most $n$ elements. Let $h(n)$ be a function from the natural numbers to the natural numbers. In \cite{Immerman81}, $\QN[h(n)]$ is defined to be the class of all properties $P$ for which there is a
\emph{uniform} sequence $\{\psi_n\}_{n\geq 1}$, of {\fo}-sentences such that $\psi_n$ expresses $P$ on structures with at most $n$ elements, $\psi_n$ has $O(h(n))$ quantifiers, and $h(n)$ is constructible in $\DSPACE[h(n)]$, where $\DSPACE$ stands for deterministic space. 
From results in \cite{Immerman81}, it follows that on the class of all ordered finite structures, 
$$\NL\subseteq \QN[\log(n)].$$
Thus, proving that some \np-complete  or some \p-complete problem requires $\omega(\log(n))$ quantifiers would separate $\NL$ from $\np$ or from $\p$, and this could be achieved by playing \ms games on ordered finite structures. In contrast, as pointed out in \cite{Immerman81}, {\fo}-sentences of quantifier rank $O(\log(n))$ capture isomorphism on ordered finite structures, and hence \ef games cannot be used to prove lower bounds for sequences of {\fo}-sentences with $\omega(\log(n))$ quantifiers. Even though the early optimism of using combinatorial games to separate complexity classes has yet to materialize, it is still worth exploring \ms games further with the aim of developing a deeper understanding of their features and potential uses.

On the face of their definitions, \ef games and \ms games differ in the following ways. First, \ef games are played on a pair $(\bA,\bB)$ of structures, while \ms games are played on a pair $(\cA,\cB)$ of \emph{sets}  of structures; second, in each round of the \ms game, Duplicator can make any number of copies of structures on the side that she has to play; and, third, in the \ef game, Duplicator must maintain a partial isomorphism between $\bA$ and $\bB$, while in the \ms game, she must maintain a partial isomorphism between some structure from $\cA$ and some structure from $\cB$. Spoiler wins the $r$-round \ef game on $(\bA,\bB)$ iff there is a {\fo}-sentence of quantifier rank $r$ that is true on $\bA$ and false on $\bB$, while Spoiler wins the $r$-round \ms game on $(\cA,\cB)$ iff there is a {\fo}-sentence with $r$ quantifiers that is \emph{separating} for $(\cA,\cB)$, i.e.,~it is  true on every structure in $\cA$ and false on every structure in $\cB$.

In this paper, we carry out a finer analysis of \ms games and shed additional light on how \ef games compare with \ms games. In particular, when applying \ms games to obtain inexpressibility results, one can use infinite sets $\cA$ and $\cB$ of structures. We show that the use of infinite sets does not yield stronger inexpressibility results; however, the sets used cannot always be assumed to be singletons (i.e.,~we cannot always start with a pair of structures). We also show that  $r$-round \ms games restricted to singleton sets capture Boolean combinations of {\fo}-sentences with $r$ quantifiers. Finally, we add insight into Duplicator's ability to make copies of structures during the game: if enough copies of each structure are made before the first round, then Duplicator does not need to make any additional copies during gameplay.

In \cite{MSgames1} and \cite{MSgames2}, it was pointed out that, unlike the \ef game, Spoiler may get an advantage in the \ms game by playing ``on top,'' i.e.,~by placing a pebble on top of an existing pebble or a constant. Here, we first give a self-contained proof of the fact that, in general, Spoiler may not win an \ms game without sometimes playing on top. We then explore the expressive power of the variant of the \ms game in which Spoiler never plays on top. By analyzing the \emph{types} (i.e.,~the combinations of atomic and negated atomic formulas satisfied by the pebbles played and the constants), we characterize when Spoiler wins the $r$-round \ms game on $(\cA,\cB)$ without playing on top in terms of properties of separating sentences for $(\cA,\cB)$.

We next study combinatorial games that simultaneously capture number of quantifiers and number of variables. It is well known that \ef games can be adapted to simultaneously capture quantifier rank and number of variables by limiting the number of pebbles used. In contrast, we show that natural variants of \ms games obtained by limiting the number of pebbles fail to simultaneously capture number of quantifiers and number of variables. For this reason, we introduce a new game, the \emph{quantifier-variable tree} game (\qvt game). The \qvt game  is inspired by the Adler-Immerman game \cite{adlerimmerman}, which was introduced to study formula size. Variants of the Adler-Immerman game were subsequently investigated in \cite{gsSuccinctness} to study the succinctness of logics, and in \cite{hv15} to obtain lower bounds for formula size in propositional logic and in {\fo}. Our main result about the \qvt game asserts that Spoiler wins the $r$-round, $k$-pebble \qvt game on a pair $(\cA,\cB)$ of sets of structures if and only if there is a separating sentence for $(\cA,\cB)$ with $r$ quantifiers and $k$ variables.

Finally, we generalize the \qvt game to a class of two-player games called \emph{syntactic games}, that simultaneously capture some other measure of first-order formulas and the number of variables. We focus on \emph{compositional} syntactic measures, i.e., measures on first-order formulas defined inductively for formulas in terms of its value on their subformulas. Quantifier count, quantifier rank, and formula size are all examples of such measures.  Our syntactic games provide a unified game-theoretic setting in which to study these measures. A similar game for formulas  of the infinitary logic $\mathcal{L}_{\omega_1, \omega}$ was studied in \cite{DBLP:journals/mlq/VaananenW13}.

In conclusion, our results yield new insights into \ms games, their variants, and generalizations. The next step in this investigation would be to use these games to 
determine the optimal value of a complexity measure and number of variables to express various combinatorial properties as a function of the size of the relevant structures.
\section{Preliminaries}\label{sec:prelims}

Throughout this paper, we consider a relational schema $\tau$ consisting of finitely many relation and constant symbols. We denote the set of relation symbols in $\tau$ by $\pred(\tau)$. We write $\tau$-structures in boldface, e.g.,~$\bA$. We denote the universe (i.e.,~set of elements) of $\bA$ by $A$, and sets of $\tau$-structures using calligraphic typeface, e.g.,~$\cA$.

\subsection{Pebbled structures}

We have a set $\cC$ of \emph{pebble colors}, with arbitrarily many \emph{pebbles} of each color available. A $\tau$-structure $\bA$ is \emph{pebbled} if zero or more elements from $A$ have one or more pebbles on them, so that \emph{at most} one pebble of each color is present in $\bA$. If pebbles of color $1, \ldots, t$ are placed on elements $a_1, \ldots, a_t \in A$, we refer to this pebbled structure as $\langle\bA ~|~ a_1, \ldots, a_t\rangle$. For readability, when the context is obvious, we refer to $\langle\bA ~|~ a_1, \ldots, a_t\rangle$ as the $t$-pebbled (or simply pebbled) structure $\bA$. Note that the unpebbled structure $\bA$ corresponds to the pebbled structure $\langle \bA ~|~ \rangle$, with an empty set of pebbles. When $t = 0$, the notation $\langle \bA ~|~ a_1, \ldots, a_t\rangle$ refers to the unpebbled structure $\langle \bA ~|~ \rangle$.

\begin{defi}\label{def:isomorphism}
    Consider two pebbled structures $\langle\bA ~|~ a_1, \ldots, a_t\rangle$ and $\langle\bB ~|~ b_1, \ldots, b_t\rangle$, and let $f: A \to B$ be the map such that:
    \begin{itemize}
    \item for $1 \leq i \leq t$, we have that $f$ maps $a_i \mapsto b_i$.
    \item for each constant symbol $c$ in $\tau$, we have that $f$ maps $c^\bA \mapsto c^\bB$.
    \end{itemize}
    We say the pair $\langle\bA ~|~ a_1, \ldots, a_t\rangle$ and $\langle\bB ~|~ b_1, \ldots, b_t\rangle$ of pebbled structures \emph{matches} (and call it a \emph{matching pair})  iff the map $f$ defined above is an isomorphism between the substructures of $\bA$ and $\bB$ induced by its domain and range.
\end{defi}

Of course, note that the pebbled structures $\langle\bA ~|~ a_1, \ldots, a_t\rangle$ and $\langle\bB ~|~ b_1, \ldots, b_t\rangle$ may form a matching pair even if $\bA$ and $\bB$ are not isomorphic. In fact, they may even form a matching pair without any pebbles placed on them.

\subsection{The \MS game}
Fix $r \in \N$, and two nonempty sets $\cA$ and $\cB$ of $\tau$-structures. We define the $r$-round \MS game on $(\cA, \cB)$ as follows. We refer to $\cA$ and $\cB$ as the \emph{left} and \emph{right} sides, respectively. Informally, in each round of the game, one player (Spoiler) places a pebble on every structure on either the left side or the right side. Then the other player (Duplicator) places a pebble of the same color on every structure on the other side. Duplicator (but not Spoiler) may copy structures to make different placements on distinct copies of the same structure. Duplicator tries to ensure that, at the end of every round, there is some matching pair (Definition \ref{def:isomorphism}); Spoiler tries to eliminate all matching pairs in $r$ or fewer rounds.

We now give a formal description of the game. We view all structures henceforth as pebbled structures. The game proceeds over up to $r$ rounds, by building a sequence of \emph{configurations} $(\cA_0, \cB_0), (\cA_1, \cB_1), \ldots$, where $\cA_0 = \cA$ and $\cB_0 = \cB$. Initially, if there are no matching pairs, then Spoiler wins the $0$-round game. Otherwise, inductively, in each round $t = 1, \ldots, r$, with configuration $(\cA_{t-1}, \cB_{t-1})$, Spoiler selects one of the following moves:
\begin{itemize}
\item $\playleft$:
  \begin{enumerate}
  \item For each pebbled structure $\langle\bA ~|~ a_1, \ldots, a_{t-1}\rangle \in \cA_{t-1}$, Spoiler places a pebble colored $t$ on an element $a_t \in A$, creating the pebbled structure $\langle\bA ~|~ a_1, \ldots, a_t\rangle$. Call the resulting set $\cA_t$.
  \item For each pebbled structure $\langle\bB ~|~ b_1, \ldots, b_{t-1}\rangle \in \cB_{t-1}$, Duplicator \emph{may} make any number of copies of this pebbled structure, and then for each such copy, \emph{must} place a pebble colored $t$ on an element $b_t \in B$, creating the pebbled structure $\langle\bB ~|~ b_1, \ldots, b_t\rangle$. Call the resulting set $\cB_t$.
  \item At the end of this round, if no $t$-pebbled structure in $\cA_t$ forms a matching pair with a $t$-pebbled structure in $\cB_t$, then Spoiler wins the game. Otherwise, if $t < r$, play continues.
   \end{enumerate}
 \item $\playright$: This move is dual to $\playleft$; Spoiler places pebbles colored $t$ on the pebbled structures on the right side, and Duplicator responds on the left side.
   \end{itemize}
At the end of round $r$, Duplicator wins the game if there is still some $r$-pebbled structure on the left and some $r$-pebbled structure on the right forming a matching pair.

After $t$ rounds of the $r$-round \ms game on $(\cA,\cB)$ have been played, we have the collection $\cA_t$ of pebbled structures $\langle\bA ~|~ a_1,\ldots,a_t\rangle$ on the left, and the collection $\cB_t$ of pebbled structures $\langle\bB ~|~ b_1,\ldots,b_t\rangle$ on the right, where $\bA$ and $\bB$ range over $\cA$ and $\cB$ respectively, and $a_i$ and $b_i$ are the elements pebbled by the players on $\bA$ and $\bB$ in round $i$, for $1\leq i\leq t$. A \emph{strategy} for Spoiler is a function that takes as input such a configuration $(\cA_t, \cB_t)$, and provides him with the next move in the game. Specifically, given such a configuration, a strategy tells Spoiler whether to play $\playleft$ or $\playright$; if it tells Spoiler to play $\playleft$, then for each pebbled structure $\langle\bA ~|~ a_1,\ldots,a_t\rangle \in \cA_t$, the strategy provides an element $a_{t+1}\in A$ for Spoiler to play the pebble colored $t+1$ on. Similarly, if it tells Spoiler to play $\playright$, then for each pebbled structure $\langle\bB ~|~ b_1,\ldots,b_t\rangle \in \cB_t$, the strategy provides an element $b_{t+1} \in B$ for Spoiler to play the pebble colored $t+1$ on. A \emph{winning} strategy for Spoiler in the $r$-round \ms game on $(\cA,\cB)$ is a strategy such that Spoiler always wins the game if he follows it.

We say that Duplicator follows the \emph{oblivious} strategy if in each round, for every pebbled structure on the side she plays on, she makes as many copies as there are elements in the universe of the structure, and then plays a pebble on a different element in each copy. It is obvious that if Duplicator wins the $r$-round \ms game on $(\cA,\cB)$, then Duplicator can win  by following the oblivious strategy. From now on, unless we say otherwise, we will assume that Duplicator follows the oblivious strategy in the \ms game.

We next state an important definition.

\begin{defi}\label{def:sepsentence}
Given two sets $\cA$ and $\cB$ of $\tau$-structures, a \emph{separating sentence for $(\cA,\cB)$} is a {\fo}-sentence
$\psi$ such that every $\bA \in \cA$ has $\bA \models \psi$, and every $\bB \in \cB$ has $\bB \models \lnot\psi$.
\end{defi}

Clearly, $\psi$ is a separating sentence for $(\cA, \cB)$ if and only if $\lnot\psi$ is a separating sentence for $(\cB,\cA)$.

Next, we state without proof the fundamental theorem of \ms games, which is in  \cite[Theorem 10]{Immerman81} and \cite[Theorem 1.2]{MSgames1}.

\begin{thmC}[\cite{Immerman81, MSgames1}]\label{thm:ms-main}
Spoiler has a winning strategy for the $r$-round \ms game on $(\cA, \cB)$ if and only if there is an $r$-quantifier separating sentence for $(\cA,\cB)$.
\end{thmC}

As discussed in \cite{Immerman81}, Duplicator's ability to make copies of the structures in the \ms game is crucial. Consider the next example (Figure \ref{fig:3_vs_2}).

\tikzset{RED/.style={draw=red, very thick, fill=none, font=\tiny, circle, inner sep=0.5mm}}
\tikzset{BLUE/.style={draw=blue, very thick, fill=none, font=\tiny, circle, inner sep=0.5mm}}
\tikzset{BLUESMALL/.style={draw=blue, very thick, fill=none, font=\tiny, circle, inner sep=0.3mm}}
\tikzset{BLANK/.style={draw=black, fill=none, font=\tiny, circle, inner sep=0.5mm}}
\tikzset{LABEL/.style={draw=none, fill=none, font=\small, circle, inner sep=0cm}}
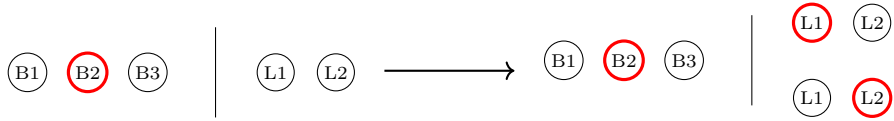
\begin{figure}[ht]
\centering
\begin{tikzpicture}

\node [BLANK] at (-2.3,0) {B1};
\node [RED] at (-1.5,0) {B2};
\node [BLANK] at (-0.7,0) {B3};

\draw (0.2, 0.6) -- (0.2, -0.6);

\node [BLANK] at (1,0) {L1};
\node [BLANK] at (1.8,0) {L2};

\end{tikzpicture}
\begin{tikzpicture}[state/.style={rectangle,draw=none, rounded corners},arrow/.style={->,thick}]
  \node[state] (tail) at (-1,0) {};
  \node[state] (head) at (1,0) {};
  \node[state] (invis) at (0,-0.5) {};
  \draw[arrow] (tail) edge node[auto] {} (head);
\end{tikzpicture}
\begin{tikzpicture}

\node [BLANK] at (-2.3,0) {B1};
\node [RED] at (-1.5,0) {B2};
\node [BLANK] at (-0.7,0) {B3};

\draw (0.2, 0.6) -- (0.2, -0.6);

\node [RED] at (1,0.5) {L1};
\node [BLANK] at (1.8,0.5) {L2};

\node [BLANK] at (1,-0.5) {L1};
\node [RED] at (1.8,-0.5) {L2};

\end{tikzpicture}
\caption{The beginning of a  $2$-round \ms game on $(\cA, \cB)$, where $\cA = \{\lo(3)$\}, the singleton linear order of size $3$, and $\cB = \{\lo(2)\}$, the singleton linear order of size $2$. Spoiler plays his first move (indicated in red) on B2, the middle element of $\lo(3)$. Duplicator then makes a copy of $\lo(2)$ and plays each of the two possible moves in response. This example is a slight variation on the examples in \cite{MSgames1,MSgames2}.}
\label{fig:3_vs_2}
\end{figure}

The game is played on $\{\lo(3)\}$, a singleton set consisting of a linear order of size $3$ on  the left, and $\{\lo(2)\}$, a singleton set consisting of a linear order of size $2$ on the right. Spoiler's only interesting move is to play the middle element B2 on the left, as indicated in red in the figure (it can be easily seen that all other moves by Spoiler lead to defeats). If this were now an \ef game, Duplicator would lose --- a response of L1 would be met with a Spoiler play on B1 in round $2$, while a  response of L2 would be met with a Spoiler play on B3 in round $2$. However, in  the  \ms game, Duplicator can make a second copy of $\lo(2)$, and play \emph{both} possible moves. It should now be obvious that it is impossible for Spoiler to win in the one remaining move, illustrating a clear difference between \ef and \ms games. Indeed, there is a separating sentence for $(\{\lo(3)\}, \{\lo(2)\})$ of quantifier rank $2$, namely, $\exists x (\exists y (y < x) \land \exists y (x < y))$, which uses three quantifiers, but no separating sentence with two quantifiers!

\subsection{Strategies}
A perusal of the proof of Theorem \ref{thm:ms-main} in \cite{Immerman81} or in \cite{MSgames1} reveals that if $\psi$ is a separating sentence for $(\cA,\cB)$ with $r$ quantifiers, then Spoiler has a winning strategy for the $r$-round \ms game on $(\cA,\cB)$ where he ``follows'' $\psi$.  To explain in more precise terms what ``following'' $\psi$ means, assume that $\psi$ is of the form  $Q_1x_1\ldots Q_rx_r\theta$, where $\theta$ is quantifier-free. If $Q_t$ is $\exists$, then Spoiler plays $\playleft$, while if $Q_t$ is $\forall$, then Spoiler plays $\playright$. Assume by induction that at the start of round $t \geq 1$, the configuration has the property that for every $\langle\bA ~|~ a_1,\ldots,a_{t-1}\rangle$ on the left and $\langle\bB ~|~ b_1,\ldots,b_{t-1}\rangle$ on the right in this configuration, we have:
\begin{align*}
\bA&\models Q_tx_t\ldots Q_rx_r\theta(a_1/x_1, \ldots, a_{t-1}/x_{t-1}) \\
\bB&\models \lnot Q_tx_t\ldots Q_rx_r\theta(b_1/x_1, \ldots, b_{t-1}/x_{t-1}).
\end{align*}
If $Q_t$ is $\exists$, then there is some $a_t\in A$ such that
$\bA\models Q_{t+1}x_{t+1}\ldots Q_rx_r\theta(a_1/x_1, \ldots, a_t/x_t)$; in this case, Spoiler picks such an element $a_t$ and plays the pebble on it. If $Q_t$ is $\forall$, then $\bB\models \exists x_t \lnot Q_{t+1}x_{t+1}\ldots Q_rx_r\theta(b_1/x_1, \ldots, b_{t-1}/x_{t-1})$, and so there is an element $b_t \in B$ such that $\bB\models \lnot Q_{t+1}x_{t+1}\ldots Q_rx_r\theta(b_1/x_1, \ldots, b_t/x_t)$; in this case, Spoiler picks such an element $b_t$ and plays the pebble on it. Since in each of these cases, more than one witness to the existential quantifier may exist, there may exist different strategies for Spoiler obtained by ``following'' $\psi$; however, each of these strategies is a winning strategy for Spoiler. We say that a strategy for Spoiler is \emph{obtained} from $\psi$ if it is one of the strategies for Spoiler obtained by following $\psi$ in this way.

The preceding discussion is summarized in the first part of the next result; the second part of the result can be extracted from the proof of Theorem \ref{thm:ms-main}. 

\begin{thm}\label{thm:m-s-game-strong}
Let $r \in \N$, and let $\cA$ and $\cB$ be two sets of $\tau$-structures. Then, the following are true.
\begin{enumerate}
\item If $\psi\equiv  Q_1x_1\ldots Q_rx_r\theta$ is a separating sentence for $(\cA,\cB)$ with $r$ quantifiers, then  Spoiler can win the $r$-round \ms game on $(\cA,\cB)$ using a strategy  obtained from $\psi$. Moreover, for every such strategy and for all pebbled structures $\langle\bA ~|~ a_1,\ldots,a_r\rangle$ on the left and $\langle\bB ~|~ b_1,\ldots,b_r\rangle$ on the right in the final configuration arising by using this strategy, we have $\bA \models \theta(a_1/x_1,\dots,a_r/x_r)$ and $\bB\models \lnot \theta(b_1/x_1,\ldots,b_r/x_r)$.
\item If Spoiler has a winning strategy $\cS$ for the $r$-round \ms game on $(\cA, \cB)$, then there is a separating sentence $\psi$ for $(\cA,\cB)$ with $r$ quantifiers such that the strategy $\cS$  is one of the strategies for Spoiler  obtained from $\psi$. 
\end{enumerate}
\end{thm}

\subsection{Types}

The following concepts will recur throughout this paper. An \emph{atomic} formula is an expression of one of the following forms: an equality $x_i = x_j$ between two distinct variables;  an equality $c_i=x_j$ between a constant symbol and a variable;  an equality $c_i=c_j$ between two distinct constant symbols; an expression $R(y_1,\ldots,y_m)$, where $R$ is an $m$-ary relation symbol in the given schema, for some $m\geq 1$, and $y_1,\ldots,y_m$ are not-necessarily-distinct variables or constant symbols. A \emph{type} over an $r$-tuple of distinct variables, $x_1, \ldots, x_r$, is a conjunction of atomic and negated atomic formulas such that for each atomic formula with variables from $x_1, \ldots, x_r$, exactly one of the atomic formula and its negation appears as a conjunct. We will often denote such a type by $t(x_1, \ldots, x_r)$. If we wish to emphasize the numbers of variables, we can also refer to such a type as an \emph{$r$-type}.
\section{\EF Games vs.~Multi-Structural Games}\label{sec:differenceswithEF}

\EF games capture indistinguishibility with respect to quantifier rank, while \MS games capture indistinguishibility with respect to number of quantifiers. Each of these two families of games yields a method for establishing lower bounds in the corresponding fragment of first-order logic. Here, we compare the two methods and uncover differences between them.

A \emph{property} of $\tau$-structures is a Boolean query on $\tau$, i.e.,~a function $P$ such that for every $\tau$-structure $\bA$,  either $P(\bA)=1$, or $P(\bB)= 0$, and $P$ is invariant under isomorphisms. If $P(\bA)=1$,  we say that $\bA$ \emph{satisfies} $P$. In the rest of the presentation, we shall say Spoiler (resp.~Duplicator) \emph{wins} a certain \ef game to mean that he (resp.~she) has a winning strategy in that game.

\begin{method}[The \ef Game Method] \label{ef-method}
Let $r \in \N$, and let $P$ be a property of $\tau$-structures. Then the following statements are equivalent.
\begin{enumerate}
\item No {\fo}$(\tau)$-sentence of quantifier rank $r$ defines $P$.
\item There are $\tau$-structures $\bA$ and $\bB$ such that $\bA$ satisfies $P$, $\bB$ does not satisfy $P$, and Duplicator wins the $r$-round \ef game on $(\bA,\bB)$.
\end{enumerate}
\end{method}

The direction (b) $\Rightarrow$ (a) captures the ``soundness'' of Method \ref{ef-method}. It follows immediately from the basic property of \ef games that if Duplicator wins the $r$-round \ef game on $(\bA, \bB)$, then every {\fo}-sentence of quantifier rank at most $r$ that is true on $\bA$ is also true on $\bB$.
The direction (a) $\Rightarrow$ (b) captures the ``completeness" of Method \ref{ef-method}. Its proof uses the fact that {\fo}-sentences of quantifier rank $r$ are closed under disjunctions, and also the following properties of the equivalence relation $\equiv^{\mathrm {qr}}_r$, where $\bA \equiv^{\mathrm {qr}}_r \bB$ means that $\bA$ and $\bB$ satisfy the same {\fo}-sentences of quantifier rank at most $r$:
 \begin{enumerate}
 \item $\equiv^{\mathrm {qr}}_r$ has finitely many equivalence classes.
 \item Each equivalence class of  $\equiv^{\mathrm {qr}}_r$ is definable by a {\fo}-sentence of quantifier rank $r$.
\end{enumerate}

We now turn our attention to the \ms game. The following result is an 
 immediate consequence of Theorem \ref{thm:ms-main} and the determinacy of the \ms game.

\begin{method}[The \ms Game Method]\label{ms-method}
Let $r \in \N$, and let $P$ be a property of $\tau$-structures. Then the following statements are equivalent.
\begin{enumerate}
\item No {\fo}-sentence with $r$ quantifiers  defines $P$.
\item Duplicator wins the $r$-round \ms game on  $({\cA}(P),{\cB}(P))$, where ${\cA}(P)$ is the set of all $\tau$-structures that satisfy $P$ and ${\cB}(P)$ is the set of all $\tau$-structures that do not.
\end{enumerate}
\end{method}

Clearly, for every property $P$, at least one of ${\cA}(P)$ and ${\cB}(P)$ is infinite. Thus, it is natural to ask: can the \ms Game Method  be used to prove inexpressibility results by considering finite sets of structures only?  

Let $r \in \N$ and consider  the number-of-quantifiers equivalence relation $\equiv^{\mathrm{qn}}_r$, where ${\bA}\equiv^{\mathrm {qn}}_r {\bB}$ means that ${\bA}$ and ${\bB}$ satisfy the same {\fo}-sentences with   $r$ quantifiers.
The following proposition contains some basic facts about $\equiv^{\mathrm{qn}}_r$.

\begin{restatable}{prop}{qnrfacts}\label{prop:quant-num-equiv}
For every $r \in \N$, the following hold.
\begin{enumerate}
\item $\equiv^{\mathrm {qn}}_r$ has finitely many equivalence classes.
\item For every $\tau$-structure $\bA$, the $\equiv^{\mathrm {qn}}_r$-equivalence class of $\bA$ is definable by the conjunction of all {\fo}-sentences with $r$ quantifiers  satisfied by $\bA$.
\item Each $\equiv^{\mathrm {qn}}_r$-equivalence class is definable by a {\fo}-sentence, but it need not be definable by a {\fo}-sentence with $r$ quantifiers.
\end{enumerate}
\end{restatable}

\begin{proof}
The equivalence relation $\equiv^{\mathrm{qn}}_r$ has finitely many equivalence classes\footnote{ Technically speaking, this is true only once we have fixed the schema $\tau$. For the sake of readability, we make this assumption throughout.} because, up to logical equivalence, there are only finitely many {\fo}-sentences with $r$ quantifiers. Recall that we made the assumption throughout that our fixed schema $\tau$ has finitely many relation symbols.

Take a structure $\bA$ and consider all {\fo}-sentences with (at most) $r$ quantifiers that are true on $\bA$. Since, up to logical equivalence, there are finitely many such sentences,  the conjunction $\theta_\bA$ of these sentences is a {\fo}-sentence. Moreover,  $\theta_\bA$ defines the $\equiv^{\mathrm{qn}}_r$-equivalence class of $\bA$, because {\fo}-sentences with $r$ quantifiers are closed under negation. Indeed, if $\bB$ satisfies $\theta_\bA$, then clearly every {\fo}-sentence with $r$ quantifiers that is true on $\bA$ is also true on $\bB$. The converse is also true, since, otherwise, we would have a {\fo}-sentence $\psi$ with $r$ quantifiers that is true on $\bB$ but not on $\bA$. But then $\neg \psi$ would be true on $\bA$ and false on $\bB$, which is a contradiction.

In general, the sentence $\theta_\bA$  need not be logically equivalent to any {\fo}-sentences with  $r$ quantifiers. To see this, consider the case $r=1$ and a relational schema consisting of a binary relation symbol $R$. Let $\bA$ be the structure consisting of a self-loop $R(a,a)$ and an isolated node $b$. Then $\bA$ satisfies $\exists xR(x,x)\land \exists y \neg R(y,y)$, but no {\fo}-sentence with a single quantifer logically implies $\exists xR(x,x)\land \exists y \neg R(y,y)$, and hence the $\equiv^{\mathrm{qn}}_r$-equivalence class of $\bA$ cannot be expressed by a {\fo}-sentence with a single quantifier. This generalizes to $r > 1$.
\end{proof}

Note also that, unlike {\fo}-sentences of quantifier rank $r$, {\fo}-sentences with $r$ quantifiers are not closed under finite disjunctions or finite conjunctions.

The next proposition shows that it suffices to consider finite sets in \ms games.

\begin{restatable}{prop}{finitesuffice}\label{prop:finite}
Let $r \in \N$, let $\cA$ and $\cB$ be two sets of $\tau$-structures, and let $\cA'$ and $\cB'$ be the sets of $\tau$-structures obtained from  $\cA$ and  $\cB$ by keeping exactly one structure from each equivalence class 
$\equiv_r^{\rm qn}$ with members in  $\cA$ and $\cB$, respectively.
Then the following hold.

\begin{enumerate}
\item The sets ${\cA}'$ and  ${\cB}'$ are finite.
\item Duplicator wins the $r$-round \ms game on
$({\cA},{\cB})$ if and only if  Duplicator wins the $r$-round \ms game on $({\cA}',{\cB}')$.
\end{enumerate}
\end{restatable}

\begin{proof}
The first part follows immediately from Proposition \ref{prop:quant-num-equiv}, by the fact that the equivalence relation $\equiv_r^{\mathrm{qn}}$ has finitely many equivalence classes. 

For the second part, since ${\cA}'\subseteq {\cA}$ and ${\cB}'\subseteq {\cB}$, if Duplicator wins the $r$-round \ms game on ${\cA}'$
and  ${\cB}'$, then Duplicator also wins the $r$-round \ms game on ${\cA}$
and  ${\cB}$. For the other direction, assume that Duplicator  wins the $r$-round \ms game on ${\cA}$
and  ${\cB}$. We have to show that Duplicator wins the $r$-round \ms game on ${\cA}'$
and  ${\cB}'$. If this were not true, then Spoiler wins the $r$-round \ms game on ${\cA}'$
and  ${\cB}'$. Hence, by Theorem \ref{thm:ms-main}, there is a {\fo}-sentence $\psi$ with $r$ quantifiers such that $\psi$ is true on every structure in ${\cA}'$ and false on every structure in ${\cB}'$. Since every structure in $\cA$ (resp.~${\cB}$) is $\equiv_r^{\mathrm{qn}}$ to a structure in ${\cA}'$ (resp.~${\cB}'$), it follows that $\psi$ is true on every structure in $\cA$ and false on every structure in $\cB$. Hence, by Theorem \ref{thm:ms-main}, Spoiler wins the $r$-round game on ${\cA}$
and ${\cB}$, which is a contradiction.
\end{proof}

By combining Method \ref{ms-method} and Proposition \ref{prop:finite},
we obtain a method that brings Method \ref{ms-method} closer to Method \ref{ef-method}.

\begin{method}[The \ms Game Method revisited]\label{ms-fin-method}
Let $r \in \N$, and let $P$ be a property of $\tau$-structures. Then the following statements are equivalent.
\begin{enumerate}
\item No {\fo}-sentence with $r$ quantifiers  defines $P$.
\item There are finite sets
$\cA$ and $\cB$ of $\tau$-structures such that every structure in $\cA$ satisfies $P$, no structure in $\cB$ satisfies $P$, and Duplicator wins the $r$-round \ms game on  $({\cA},{\cB})$.
\end{enumerate}
\end{method}

 Can Method \ref{ms-fin-method} be further refined so that inexpressibility results about {\fo}-sentences with $r$ quantifiers are proved by playing \ms games on pairs of \emph{singleton} sets? It is clear that if Spoiler wins the $r$-round \ms game on $(\cA,\cB)$, then for every $\bA \in \cA$ and every $\bB \in \cB$, Spoiler wins the $r$-round \ms game on $(\{\bA\},\{\bB\})$. The next example shows that the converse need not be true, even if $r=1$.

 \begin{exa}\label{ex:notsingletons}
Consider a schema with two unary predicates $R$ and $G$, and consider the structures $\bA$, $\bB$, and $\bC$, where
\begin{itemize}
\item $\bA$ has $\{a_1,a_2\}$ as its universe, $R^{\bA} = \{a_1\}$, $G^{\bA} = \{a_2\}$.
\item $\bB$ has $\{b\}$ as its universe, $R^{\bB} = \{b\}$, $G^{\bB} = \varnothing$.
\item $\bC$ has $\{c\}$ as its universe, $R^{\bC} = \varnothing$, $G^{\bC} = \{c\}$.
\end{itemize}
It is easy to check that Spoiler wins  
  the $1$-round \ms game on 
$(\{{\bA}\},\{{\bB}\})$ and also on  $(\{\bA\},\{{\bC}\})$, but Duplicator wins the $1$-round \ms game on 
 $(\{{\bA}\}, \{{\bB}, {\bC}\})$. 
\end{exa}
Thus, Duplicator may win the \ms game on two sets of structures, but  Duplicator may \emph{not} win the \ms game on any pair of structures from these two sets.

What do \ms games restricted to singleton sets of structures capture? The next result gives the answer.

\begin{thm}\label{thm:MS-bool-comb}
Let $r \in \N$, and let $P$ be a property of $\tau$-structures. Then the following statements are equivalent.
\begin{enumerate}
\item The property $P$ is definable by a Boolean combination of {\fo}-sentences each with $r$ quantifiers.
\item For all $\tau$-structures $\bA$ and $\bB$, if $\bA$ satisfies property $P$ and Duplicator wins the $r$-round \ms game on 
$(\{\bA\},\{\bB\})$, then $\bB$ satisfies property $P$.
\end{enumerate}
\end{thm}

\begin{proof}
Assume that $P$ is a property definable by a Boolean combination $\psi$ of {\fo}-sentences each of which has  $r$ quantifiers. WLOG assume that $\psi$ is in disjunctive normal form, i.e.,~$\psi$ is a disjunction $\bigvee_{i=1}^m\theta_i$, where each $\theta_i$ is a conjunction of {\fo}-sentences each with  $r$ quantifiers (observe that the negation of a {\fo}-sentence with $r$ quantifiers is a {\fo}-sentence with $r$ quantifiers). Assume now that $\bA$ is a structure that satisfies $P$ and that  Duplicator wins the $r$-round \ms game on  $(\{\bA\},\{\bB\})$. Since $\bA$ satisfies  $P$,  there is some $i\leq m$ such that $\bA \models \theta_i$. Since  Duplicator wins the $r$-round \ms game on  $(\{\bA\},\{\bB\})$, Theorem \ref{thm:ms-main} implies that every {\fo}-sentence with $r$ quantifiers that is true on $\bA$ is also true on $\bB$, so that every conjunct of $\theta_i$ is true on $\bB$, and hence $\bB\models \theta_i$. This implies that $\bB\models \psi$ and so $\bB$ satisfies  $P$.

Conversely, assume that for all structures $\bA$ and $\bB$, if $\bA$ satisfies  $P$ and  Duplicator wins the $r$-round \ms game on $(\{\bA\},\{\bB\})$, then $\bB$ satisfies  $P$. For every structure $\bA$ satisfying $P$, let $\theta_\bA$ be the 
conjunction of all {\fo}-sentences with $r$ quantifiers  satisfied by $\bA$. By Proposition \ref{prop:quant-num-equiv}, $\theta_\bA$ is a {\fo}-sentence that defines the $\equiv_r^{\mathrm {qn}}$-equivalence class of $\bA$. Let $\psi$ be the (finite) disjunction
$\bigvee_{\bA\models P}\theta_\bA$. Clearly, $\psi$ is a Boolean combination of {\fo}-sentences each with $r$ quantifiers.
We claim that $\psi$ defines  $P$. Assume first that $\bB$ is a structure that satisfies $P$. Then $\bB \models \theta_\bB$, and $\theta_\bB$ is a disjunct  of $\psi$, hence $\bB\models \psi$. Conversely, assume that $\bB$ is a structure such that $\bB \models \psi$. Hence, there is a structure $\bA$ such that $\bA$ satisfies $P$ and $\bB\models \theta_\bA$. Since $\theta_\bA$ defines the $\equiv_r^{\mathrm {qn}}$-equivalence class of $\bA$, it follows that $\bA$ and $\bB$ satisfy the same {\fo}-sentences with $r$ quantifiers. Therefore, by Theorem \ref{thm:ms-main}, we have that  Duplicator wins the $r$-round \ms game on $(\{\bA\},\{\bB\})$,  hence, by our hypothesis, we have that $\bB$ satisfies $P$.
\end{proof}

As an immediate consequence of Theorem \ref{thm:MS-bool-comb}, we obtain the following  method about \ms games on singleton sets.

\begin{method} \label{ms-boolean-method} 
Let $r \in \N$, and let $P$ be a property of $\tau$-structures. Then the following statements are equivalent.

\begin{enumerate}
\item There is no Boolean combination of {\fo}-sentences each with $r$ quantifiers that defines $P$.
\item There are $\tau$-structures $\bA$ and $\bB$ such that  $\bA$ satisfies $P$, $\bB$ does not satisfy $P$, and Duplicator wins the $r$-round \ms game on $(\{\bA\},\{\bB\})$.
\end{enumerate}
\end{method}

Because of the disjunctive normal form and since {\fo}-sentences with $r$ quantifiers are closed under negation, we could replace ``Boolean combination" by ``disjunction of conjunctions" in the statements of Theorem \ref{thm:MS-bool-comb} and Method \ref{ms-boolean-method}.

In summary, we now have the following rather complete picture about the similarities and the differences between \ef games and \ms games:
\begin{enumerate}
\item  $r$-round \ef games played on a pair of structures capture expressibility by a {\fo}-sentence of quantifier rank $r$.

\item  $r$-round \ms games played on sets of structures capture expressibility by a {\fo}-sentence with $r$ quantifiers.

The sets of structures can be taken to be finite, but not singletons.

\item  $r$-round \ms games played on a pair of structures (i.e.,~on singletons) capture expressibility by a Boolean combination of {\fo}-sentences, each with $r$ quantifiers.
\end{enumerate}
The difference between (a) and (c) is because {\fo}-sentences of quantifier rank $r$ are closed under disjunctions and conjunctions, whereas {\fo}-sentences with $r$ quantifiers are not.

We also have the following  trade-off between the \ef Game Method and the \ms Game Method in establishing
inexpressibility results about a property $P$: for \ef games, we can work with pairs of structures, but we need to find them and to show that they have the desired properties in Method \ref{ef-method}; in contrast, for \ms games we can collect all the structures that satisfy property $P$ on one side and all the structures that do not satisfy $P$ on the other, but then it may be  harder to show that Duplicator wins the \ms game on these sets.

\medskip

Duplicator's ``superpower'' in \ms games arises from the ability to make copies of structures while the game is played. The next theorem says that, if enough copies of each structure are made at the beginning of the game, then Duplicator does not need to  make any further copies during gameplay.

\begin{restatable}{thm}{extracopies}\label{thm:dupl}
Fix $r \in \mathbb{N}$, and let $\cA$ and $\cB$ be two sets of finite $\tau$-structures. Then, the following are equivalent.
\begin{enumerate}
\item Duplicator wins the  $r$-round \ms game on $(\cA,\cB)$.
\item Duplicator wins the  $r$-round \ms game \textbf{without} duplicating on $(\mathcal{A}^+, \mathcal{B}^+)$, where 
 $\mathcal{A}^+$ is the multiset consisting of $|A|^r$ copies of each $\bA \in \mathcal{A}$, and $\mathcal{B}^+$ is the multiset consisting  of $|B|^r$ copies of each  $\bB \in \mathcal{B}$.
 \end{enumerate}
 \end{restatable}

\begin{proof}
Fix a number $r$ of rounds. Let us refer to the ordinary MS game on $(\cA, \cB)$ as the \emph{original} game, and to the MS  game on $(\cA^+, \cB^+)$ where Duplicator is not allowed to make copies as the \emph{new} game. Recall that the game proceeds by building a sequence of configurations $(\cA_0, \cB_0), (\cA_1, \cB_1), \ldots$, where $\cA_0 = \cA$ and $\cB_0 = \cB$. In particular, the configuration at the beginning of round $t$ for $1 \leq t \leq r$ is $(\cA_{t-1}, \cB_{t-1})$, consisting of $(t-1)$-pebbled structures (including when $t = 1$). We define the configurations $(\cA_t^+, \cB_t^+)$ of the new game analogously.

\medskip

\noindent\textbf{Claim 1.} If Spoiler has a winning strategy in the original game on $(\cA, \cB)$, then Spoiler has a winning strategy in the new game on $(\cA^+, \cB^+)$.

\medskip

This direction is quite easy: if Spoiler has a winning strategy in the original game on $(\cA, \cB)$, then by the Fundamental Theorem of MS Games, this means there is some separating FO sentence $\varphi \in \mathcal{L}(\tau)$ such that $\varphi$ is true of all $\bA \in \cA$, and $\varphi$ is false of all $\bB \in \cB$. Since the \emph{set} of structures in $\cA^+$ (resp.~$\cB^+$) is precisely the set of structures in $\cA$ (resp.~$\cB$), this must also mean that $\varphi$ is a separating sentence for $(\cA^+, \cB^+)$ as well. Therefore, Spoiler has a winning strategy in $(\cA^+, \cB^+)$ for the original game. Then, since the only difference between the old and new games is that Duplicator is more restricted in the new game, Spoiler certainly has a winning strategy in the new game for $(\cA^+, \cB^+)$ as well.

For an alternative argument involving an explicit description of Spoiler's winning strategy, we proceed as follows.

Fix an arbitrary instance $(\cA, \cB)$ of the original game, and assume that Spoiler has a winning strategy on this instance in the original game (and WLOG, assume that Duplicator plays obliviously, so that Spoiler's strategy is \emph{non-adaptive} in the following sense: his sequence of moves in the game is determined at the start of the game, before any move has been played). Consider the corresponding instance $(\cA^+, \cB^+)$ of the new game. We shall define a strategy for Spoiler to win this instance $(\cA^+, \cB^+)$ in the new game.

Suppose WLOG in round $t$, Spoiler plays in $\cA_{t-1}$ in the original game; this is well-defined since Spoiler's strategy is non-adaptive. We shall ensure that in round $t$ of the new game, Spoiler plays in $\cA_{t-1}^+$. A symmetric argument will hold if Spoiler plays in $\cB_{t-1}$.

To lift Spoiler's move in round $t$ to a move in the new game, we will ensure that the following invariant holds for all $t$ satisfying $1 \leq t \leq r + 1$:

\medskip

\noindent\emph{Invariant 1.} Every $\langle\bA ~|~ a_1, \ldots, a_{t-1}\rangle \in \cA_{t-1}^+$ in the new game is present in $\cA_{t-1}$ in the original game. Similarly, every $\langle\bB ~|~ b_1, \ldots, b_{t-1}\rangle \in \cB_{t-1}^+$ in the new game is present in $\cB_{t-1}$ in the original game.

\medskip

As long as Invariant 1 holds, Spoiler can take \emph{all} copies of any particular pebbled structure in the new game, and play the move he plays on that structure in the original game; this would be a well-specified strategy.

To see that Invariant 1 holds under this strategy, we shall use induction on $t$. Clearly, it is true for $t = 1$, by the definition of $\cA^+$ and $\cB^+$. Suppose Invariant 1 is true of $t - 1$, and assume Spoiler follows the strategy above in round $t$.

Consider a pebbled structure $\langle\bA ~|~ a_1, \ldots, a_t\rangle \in \cA_t^+$ in the new game. This structure arose because in round $t$, Spoiler played pebble $t$ on element $a_t \in A$ in the pebbled structure $\langle\bA ~|~ a_1, \ldots, a_{t-1}\rangle$. Therefore, the structure $\langle\bA ~|~ a_1, \ldots, a_{t-1}\rangle$ was in $\cA_{t-1}^+$. By induction, this structure $\langle\bA ~|~ a_1, \ldots, a_{t-1}\rangle$ was also in $\cA_{t-1}$ in the original game. By our specified strategy, Spoiler played pebble $t$ on $a_t \in A$ in the new game precisely because he also made that same move in the original game. Therefore, the structure $\langle\bA ~|~ a_1, \ldots, a_t\rangle$ exists at the end of round $t$ in the original game (and hence is in $\cA_t$).

Now consider a pebbled structure $\langle\bB ~|~ b_1, \ldots, b_t\rangle \in \cB_t^+$ in the new game. This structure arose in round $t$ when Duplicator played pebble $t$ on element $b_t \in B$ on the pebbled structure $\langle\bB ~|~ b_1, \ldots, b_{t-1}\rangle$. Hence, $\langle\bB ~|~ b_1, \ldots, b_{t-1}\rangle$ is in $\cB_{t-1}^+$. By induction, this pebbled structure $\langle\bB ~|~ b_1, \ldots, b_{t-1}\rangle$ was also in $\cB_{t-1}$ in the original game. Since Spoiler plays his round $t$ move in $\cA_{t-1}$, we know Duplicator moves obliviously in $\cB_{t-1}$ in the original game, and makes all possible moves, including a move that plays pebble $t$ on $b_t \in B$. (Note that we make no assumptions about Duplicator's strategy in the \emph{new} game.) Hence, the pebbled structure $\langle\bB ~|~ b_1, \ldots, b_t\rangle$ arises during round $t$ in the original game (and hence is in $\cB_t$).

Therefore, Invariant 1 is true of $t$, and so by induction, Spoiler's strategy is well-specified.

By assumption, we used a winning strategy for Spoiler in the original game on $(\cA, \cB)$; therefore, at the end of $r$ rounds in the original game, no pair of pebbled structures from $\cA_r$ and $\cB_r$ forms a matching pair. Suppose in the new game there are pebbled structures $\langle\bA ~|~ a_1, \ldots, a_r\rangle \in \cA_r^+$ and $\langle\bB ~|~ b_1, \ldots, b_r\rangle \in \cB_r^+$ that form a matching pair. By Invariant 1 above, these also appear in $\cA_r$ and $\cB_r$ in the original game, which is a contradiction. Therefore, after $r$ rounds, no pair of pebbled structures from $\cA_r^+$ and $\cB_r^+$ form a matching pair, and hence Spoiler wins the new game on $(\cA^+, \cB^+)$.

\bigskip

\noindent\textbf{Claim 2.} If Duplicator has a winning strategy in the original game on $(\cA, \cB)$, then she has a winning strategy in the new game on $(\cA^+, \cB^+)$.

\medskip

Fix an arbitrary instance $(\cA, \cB)$ of the original game, and consider the corresponding instance $(\cA^+, \cB^+)$ of the new game. We shall define a strategy for Duplicator in this instance $(\cA^+, \cB^+)$ of the new game, and prove that it is well-specified. Once we do this, we shall show that this strategy in the new game essentially ``simulates'' the oblivious strategy in the original instance $(\cA, \cB)$ in the original game; therefore, if Duplicator has a winning strategy in the original instance, then her oblivious strategy would win, and we would then be able to conclude that she wins in the new game as well.

To describe Duplicator's strategy, we need to describe her move in round $t$ of the new game, for each $t$ satisfying $1 \leq t \leq r$. For this, we will ensure that the following invariant holds for all $t$ satisfying $0 \leq t \leq r + 1$:

\medskip

\noindent\emph{Invariant 2.} For each $\langle\bA ~|~ a_1, \ldots, a_t\rangle \in \cA_t^+$, there are $|A|^{r - t}$ copies of it in $\cA_t^+$. Similarly, for each $\langle\bB ~|~ b_1, \ldots, b_t\rangle \in \cB_t^+$, there are $|B|^{r - t}$ copies of it in $\cB_t^+$.

\medskip

Note that Invariant 2 is certainly true for $t = 0$, by the definition of $\cA^+$ and $\cB^+$. We proceed inductively, so suppose Invariant 2 is true of $t - 1$. Assume WLOG that in round $t$ of the new game, Spoiler plays in $\cA_{t-1}^+$. For each $\langle\bA ~|~ a_1, \ldots, a_{t-1}\rangle \in \cA_{t-1}^+$, Spoiler chooses some $a_t \in A$, and plays pebble $t$ on $a_t$, creating $\langle\bA ~|~ a_1, \ldots, a_t\rangle \in \cA_t^+$. By induction, there are $|A|^{r - t + 1}$ copies of $\langle\bA ~|~ a_1, \ldots, a_{t-1}\rangle$, but at most $|A|$ distinct new pebbled structures that can be created as a result of playing pebble $t$. So, by the pigeonhole principle, there must be some $t$-pebbled structure $\langle\bA ~|~ a_1, \ldots, a_t\rangle$ arising from $\langle\bA ~|~ a_1, \ldots, a_{t-1}\rangle$ that has at least $|A|^{r - t}$ copies after this move. We call this move by Spoiler (i.e., playing pebble $t$ on $a_t \in A$) his \emph{favorite move} on the pebbled structure $\langle\bA ~|~ a_1, \ldots, a_{t-1}\rangle \in \cA_{t-1}^+$. Duplicator keeps these $|A|^{r-t}$ copies of $\langle\bA ~|~ a_1, \ldots, a_t\rangle$ in $\cA_t^+$ (breaking ties arbitrarily in case there are multiple such structures with at least that many copies), and deletes all other $t$-pebbled structures arising from $\langle\bA ~|~ a_1, \ldots, a_{t-1}\rangle$. Note that deleting structures only makes it harder for Duplicator to win. This ensures that each set of $|A|^{r - t + 1}$ copies of $\langle\bA ~|~ a_1, \ldots, a_{t-1}\rangle \in \cA_{t-1}^+$ gives rise to $|A|^{r-t}$ copies of \emph{some} $\langle\bA ~|~ a_1, \ldots, a_t\rangle \in \cA_t^+$ (as a result of Spoiler playing his favorite move on those pebbled structures), and these are the only copies arising from that set.

Duplicator also has to respond in $\cB_{t-1}^+$. By induction, each structure $\langle\bB ~|~ b_1, \ldots, b_{t-1}\rangle \in \cB_{t-1}^+$ has $|B|^{r - t + 1}$ copies. For each $b_t \in B$, Duplicator plays pebble $t$ on element $b_t$ in exactly $|B|^{r - t}$ of these copies. This ensures that each set of $|B|^{r - t + 1}$ copies of $\langle\bB ~|~ b_1, \ldots, b_{t-1}\rangle \in \cB_{t-1}^+$ gives rise to $|B|$ new sets of $t$-pebbled structures in $\cB_t^+$, each consisting of $|B|^{r - t}$ copies of some $\langle\bB ~|~ b_1, \ldots, b_t\rangle \in \cB_t^+$ (and each corresponding to a distinct choice of $b_t$), and these are the only copies arising from that set.

Therefore, Invariant 2 is true of $t$; by induction, Duplicator's strategy is well-specified.

We shall now prove that, if Duplicator has a winning strategy for $(\cA, \cB)$ in the original game, then her strategy described above is winning in the new game for $(\cA^+, \cB^+)$. Assume henceforth that Duplicator has a winning strategy for $(\cA, \cB)$ in the original game. In particular, her oblivious strategy wins in the original game.

Consider \emph{any} sequence of moves by Spoiler in the new game on $(\cA^+, \cB^+)$, followed by Duplicator's responses as specified by her strategy above. This describes a complete run of the new game over $r$ rounds. We shall now lift this to a complete run of the original game over $r$ rounds. To do so, we shall maintain the following invariant for $0 \leq t \leq r$:

\medskip

\noindent\emph{Invariant 3.} The multiset $\cA_t^+$ contains $|A|^{r - t}$ copies of a pebbled structure $\langle\bA ~|~ a_1, \ldots, a_t\rangle$ if and only if $\cA_t$ contains one copy of the same structure (and similarly for $\cB_t^+$ and $\cB_t$).

\medskip

Note that Invariant 3 is certainly true for $i = 0$, by the definition of $\cA^+$ and $\cB^+$. We proceed inductively, so suppose Invariant 3 is true of $t - 1$. Assume WLOG that in round $t$ of the new game, Spoiler plays in $\cA_{t-1}^+$. We know from Invariant 2 that each $\langle\bA ~|~ a_1, \ldots, a_{t-1}\rangle \in \cA_{t-1}^+$ comes in a set of $|A|^{r - t + 1}$ copies. By induction, there is one copy of $\langle\bA ~|~ a_1, \ldots, a_{t-1}\rangle$ in $\cA_{t-1}$ in the original game. We shall specify Spoiler's move on this copy in the original game, in order to construct our valid run. By Duplicator's specified strategy in the new game, this set of $|A|^{r - t + 1}$ copies gives rise (through Spoiler's favorite move) to a set of $|A|^{r - t}$ copies of some $\langle\bA ~|~ a_1, \ldots, a_t\rangle \in \cA_t^+$, with all other copies deleted. In the original game, we shall have Spoiler play his favorite move on the copy as well; namely, we shall have him play pebble $t$ on element $a_t \in A$, creating $\langle\bA ~|~ a_1, \ldots, a_t\rangle \in \cA_t$. We repeat this process for all other structures in $\cA_{t-1}$, and hence specify Spoiler's move in round $t$ completely. We have Duplicator respond \emph{obliviously} in $\cB_{t-1}$ in round $t$ in the original game. By Invariant 2, we know each $\langle\bB ~|~ b_1, \ldots, b_{t-1}\rangle \in \cB_{t-1}^+$ comes in a set of $|B|^{r - t + 1}$ copies, which gives rise to $|B|$ sets each comprised of $|B|^{r - t}$ copies of $\langle\bB ~|~ b_1, \ldots, b_t\rangle \in \cB_t^+$ (each such set corresponding to a distinct choice of $b_t$). By induction, there was one copy of $\langle\bB ~|~ b_1, \ldots, b_{t-1}\rangle$ in $\cB_{t-1}$. When Duplicator responds obliviously, she creates one copy of each $\langle\bB ~|~ b_1, \ldots, b_t\rangle \in \cB_t$ in the original game (one copy for each distinct choice for $b_t$). Therefore, each such copy rising in the original game corresponds to one of the sets of $|B|^{r - t}$ copies in the new game. This completely specifies Duplicator's move in round $t$ completely, and also proves that Invariant 3 is true of $t$ as well, finishing the induction.

This specifies one particular run of the original game on $(\cA, \cB)$, where in each round, Spoiler plays his favorite move for that round on all structures in his side, and Duplicator responds obliviously on the other side. Since by assumption, this is a winning instance for Duplicator, this means that her oblivious strategy wins against any Spoiler strategy, and so in particular on this run of the game composed of Spoiler's sequence of favorite moves as well. In particular, there is some $\langle\bA ~|~ a_1, \ldots, a_r\rangle \in \cA_r$ and some $\langle\bB ~|~ b_1, \ldots, b_r\rangle \in \cB_r$ forming a matching pair. By Invariant 3, this matching pair is present in $(\cA_r^+, \cB_r^+)$ as well, and so Duplicator wins this run of the new game.

Since Spoiler's moves in the new game were arbitrary, it follows that Duplicator has a winning strategy in the new game, and we are done.
\end{proof}
\section{Play-on-Top Moves in \ms Games}\label{sec:playontop}

In \cite{MSgames1}, the authors draw attention to a surprising difference between the \ef game and the \ms game. Specifically, in the \ef game, Spoiler gains no advantage by ever placing a pebble on top of an existing pebble or a constant. As it turns out, there are instances of the \ms game where Spoiler wins, but only if he may sometimes play ``on top'', i.e.,~he places a pebble on a constant or on an existing pebble. In this section, we delineate the expressive power of the variant of the \ms game in which Spoiler never plays on top.

We begin with an example which will be used later, in which Spoiler wins the $3$-round \ms game without playing on top. For $s=3,4$, let $\lo(s)$ be the linear order of size $s$, and consider the sentence  
\begin{equation}
    \Phi_{3,\forall}: \forall x_1 \exists x_2 \exists x_3(x_1 < x_2 < x_3 \lor x_2 < x_3 < x_1).
\end{equation}
This sentence says that every element  has two smaller elements or two larger elements, and so $\lo(4)\models \Phi_{3,\forall}$, but $\lo(3)\models \lnot\Phi_{3,\forall}$. Since $\Phi_{3,\forall}$ is a separating sentence for $(\{\lo(4)\},\{\lo(3)\})$, Spoiler wins the $3$-round \ms game on $(\{\lo(4)\},\{\lo(3)\})$. Moreover, it is easy to verify that, by following the sentence $\Phi_{3,\forall}$, Spoiler can win the \ms game on $(\{\lo(4)\},\{\lo(3)\})$ without ever playing on top.

Ideas and results in \cite{MSgames1} and \cite{MSgames2} yield examples where Spoiler wins, but only by playing on top. We give a self-contained proof that Spoiler sometimes needs to play on top, using an example that involves smaller structures and only three rounds (which is the fewest rounds for which playing on top may make a difference if there are no constants in $\tau$).

\tikzset{RED/.style={draw=red, very thick, fill=none, font=\tiny, circle, inner sep=0.5mm}}
\tikzset{BLUE/.style={draw=blue, very thick, fill=none, font=\tiny, circle, inner sep=0.5mm}}
\tikzset{BLUESMALL/.style={draw=blue, very thick, fill=none, font=\tiny, circle, inner sep=0.3mm}}
\tikzset{BLANK/.style={draw=black, fill=none, font=\tiny, circle, inner sep=0.5mm}}
\tikzset{LABEL/.style={draw=none, fill=none, font=\small, circle, inner sep=0cm}}
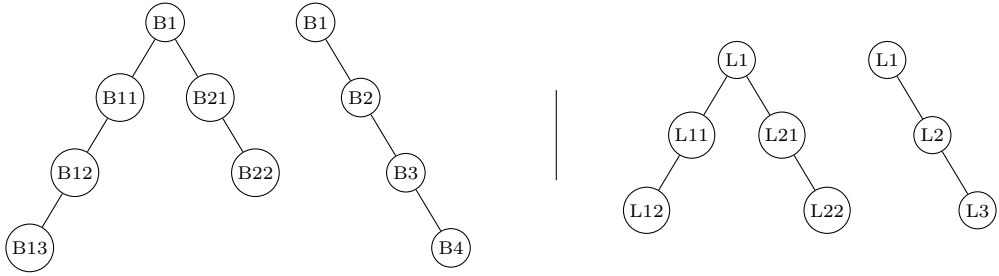
\begin{figure}[ht]
\centering
\begin{tikzpicture}

\node(L1a) [BLANK] at (-6.8,-1.5) {B13};
\node(L1b) [BLANK] at (-6.2,-0.5) {B12};
\node(L1c) [BLANK] at (-5.6,0.5) {B11};
\node(L1d) [BLANK] at (-5,1.5) {B1};
\node(L1e) [BLANK] at (-4.4,0.5) {B21};
\node(L1f) [BLANK] at (-3.8,-0.5) {B22};
\draw (L1a) -- (L1b);
\draw (L1b) -- (L1c);
\draw (L1c) -- (L1d);
\draw (L1d) -- (L1e);
\draw (L1e) -- (L1f);

\node(L2a) [BLANK] at (-3,1.5) {B1};
\node(L2b) [BLANK] at (-2.4,0.5) {B2};
\node(L2c) [BLANK] at (-1.8,-0.5) {B3};
\node(L2d) [BLANK] at (-1.2,-1.5) {B4};
\draw (L2a) -- (L2b);
\draw (L2b) -- (L2c);
\draw (L2c) -- (L2d);

\draw (0.2, 0.6) -- (0.2, -0.6);

\node(R1a) [BLANK] at (1.4,-1) {L12};
\node(R1b) [BLANK] at (2,0) {L11};
\node(R1c) [BLANK] at (2.6,1) {L1};
\node(R1d) [BLANK] at (3.2,0) {L21};
\node(R1e) [BLANK] at (3.8,-1) {L22};

\draw (R1a) -- (R1b);
\draw (R1b) -- (R1c);
\draw (R1c) -- (R1d);
\draw (R1d) -- (R1e);

\node(R2a) [BLANK] at (4.6,1) {L1};
\node(R2b) [BLANK] at (5.2,0) {L2};
\node(R2c) [BLANK] at (5.8,-1) {L3};

\draw (R2a) -- (R2b);
\draw (R2b) -- (R2c);

\end{tikzpicture}
\caption{A rooted tree $\rt(4)$ whose longest branch has $4$ nodes, and a linear order $\lo(4)$ of $4$ nodes, drawn as a rooted tree (left); a rooted tree $\rt(3)$ whose longest branch has $3$ nodes  and a linear order $\lo(3)$ of $3$ nodes  (right).}
\label{fig:2_trees1}
\end{figure}

See Figure \ref{fig:2_trees1}, whose left side contains $\rt(4)$, a rooted tree whose longest branch has $4$ nodes and, $\lo(4)$.  The right side contains $\rt(3)$, whose longest branch has $3$ nodes, and $\lo(3)$. Let $\tau$ have the single binary relation symbol $<$, where $x<y$ means that $x$ is a descendent of $y$ in the trees as drawn, so that larger nodes consistently appear higher on the page.

The following proposition shows that Spoiler needs to play on top to win this instance in the $3$-round \ms game.

\begin{restatable}{prop}{playontopnecessary}\label{prop:rooted}
Consider the sets $\cA=\{\rt(4), \lo(4)\}$ and $\cB=\{\rt(3), \lo(3)\}$ as depicted in Figure \ref{fig:2_trees}. Spoiler wins the $3$-round \ms game on $(\cA,\cB)$, but cannot win this game without playing on top.
\end{restatable}

\begin{proof}
Consider the {\fo}-sentence
\begin{eqnarray}
    \Phi_{3,\exists}: \exists x_1 \forall x_2 \exists x_3&\Big(&x_2<x_1 \rightarrow x_3 > x_1 ~ \land \notag\\
    &&x_2 >x_1 \rightarrow (x_3 \neq x_2 \land x_3 > x_1) ~ \land \notag\\
    &&x_2=x_1 \rightarrow x_3<x_1 ~~\Big),
\end{eqnarray}
which asserts that there is an element with one smaller element and two larger elements. Clearly, $\Phi_{3,\exists}$ has $3$ quantifiers and is a separating sentence for $(\cA,\cB)$; hence Theorem \ref{thm:ms-main} implies that Spoiler wins the $3$-round MS game on $(\cA,\cB)$. We will show that every winning strategy for Spoiler in the $3$-round MS game requires that Spoiler plays on top. As a stepping stone, we will first establish the following claim.

\smallskip

\noindent{\bf Claim 1:}
For every winning strategy for Spoiler  in the $3$-round \ms game on $(\cA, \cB)$, the following hold:
\begin{enumerate}
\item Spoiler's round $1$ move must be $\playleft$.
\item Spoiler's round $1$ play on $\rt(4)$ must be on either B11 or B12; furthermore, Spoiler's round $1$ play on $\lo(4)$ must be on either B2 or B3.
\item Spoiler's round $2$ move must be $\playright$.
\item Spoiler's round $3$ move must be $\playleft$.
\end{enumerate}

\smallskip

Note that (a), (c), and  (d) also  mean that every $3$-quantifier separating sentence for $(\cA,\cB)$ must have $\exists\forall\exists$ as its quantifier prefix.

\begin{proof}[Proof of Claim 1]
    \begin{enumerate}
    \item We begin by showing that Spoiler's first-round move must be on the left side. For this, it suffices to show that if Spoiler's first-round move is on the right side, then Duplicator can maintain a matching pair with one copy of $\rt(4)$ and one copy of $\rt(3)$ for three rounds. Indeed, in response to a first-round move by Spoiler on L1 in $\rt(3)$,  Duplicator plays on B1 in one copy of $\rt(4)$, and then easily survives two more rounds just on this pair of pebbled structures (e.g.,~in response to a second-round Spoiler play on B12, Duplicator plays on L11 in one copy and L12 in another copy of the $\rt(3)$). First-round Spoiler plays on L11 or L21 are met by a Duplicator play on B21 in one copy of $\rt(4)$, and first-round Spoiler plays on L12 or L22 are met by a Duplicator play on B22 in a copy of $\rt(4)$; in each case, Duplicator readily survives two more rounds on just this pair of structures. Thus, Spoiler's round $1$ move must be $\playleft$.
    \item If Spoiler's first-round move on $\rt(4)$ is not on B11 or on B12, then once again, Duplicator survives three rounds just on the pair $(\{\rt(4)\}, \{\rt(3)\})$ by keeping a matching pair between two copies via the following responses in a copy of $\rt(3)$: B1$\hookrightarrow$L1, B13$\hookrightarrow$L12, B21$\hookrightarrow$L21, B22$\hookrightarrow$L22.
    
    Furthermore, if Spoiler's first-round move on $\lo(4)$ is on B1 or B4, Duplicator survives three rounds just on the pair $(\{\lo(4)\}, \{\lo(3)\})$ by responding on L1 or L3 respectively. Therefore, Spoiler's first-round move on $\lo(4)$ must be on B2 or B3.
    \item Figure \ref{fig:after-first-round} now depicts the state of the game after the first round (with Spoiler having played on B12 on the $\rt(4)$ and on B3 on the $\lo(4)$), showing Duplicator's oblivious response on the right side as well.

\tikzset{RED/.style={draw=red, very thick, fill=none, font=\tiny, circle, inner sep=0.5mm}}
\tikzset{BLUE/.style={draw=blue, very thick, fill=none, font=\tiny, circle, inner sep=0.5mm}}
\tikzset{BLUESMALL/.style={draw=blue, very thick, fill=none, font=\tiny, circle, inner sep=0.3mm}}
\tikzset{BLANK/.style={draw=black, fill=none, font=\tiny, circle, inner sep=0.5mm}}
\tikzset{LABEL/.style={draw=none, fill=none, font=\small, circle, inner sep=0cm}}
\begin{figure}[ht]
\centering
\begin{tikzpicture}

\node(L1a) [BLANK] at (-6.8,-1.5) {B13};
\node(L1b) [RED] at (-6.2,-0.5) {B12};
\node(L1c) [BLANK] at (-5.6,0.5) {B11};
\node(L1d) [BLANK] at (-5,1.5) {B1};
\node(L1e) [BLANK] at (-4.4,0.5) {B21};
\node(L1f) [BLANK] at (-3.8,-0.5) {B22};
\draw (L1a) -- (L1b);
\draw (L1b) -- (L1c);
\draw (L1c) -- (L1d);
\draw (L1d) -- (L1e);
\draw (L1e) -- (L1f);

\node(L2a) [BLANK] at (-3,1.5) {B1};
\node(L2b) [BLANK] at (-2.4,0.5) {B2};
\node(L2c) [RED] at (-1.8,-0.5) {B3};
\node(L2d) [BLANK] at (-1.2,-1.5) {B4};
\draw (L2a) -- (L2b);
\draw (L2b) -- (L2c);
\draw (L2c) -- (L2d);

\draw (0.2, 0.6) -- (0.2, -0.6);

\node(R1a) [BLANK] at (1.4,3) {L12};
\node(R1b) [BLANK] at (2,4) {L11};
\node(R1c) [RED] at (2.6,5) {L1};
\node(R1d) [BLANK] at (3.2,4) {L21};
\node(R1e) [BLANK] at (3.8,3) {L22};

\draw (R1a) -- (R1b);
\draw (R1b) -- (R1c);
\draw (R1c) -- (R1d);
\draw (R1d) -- (R1e);

\node(R2a) [BLANK] at (1.4,1) {L12};
\node(R2b) [RED] at (2,2) {L11};
\node(R2c) [BLANK] at (2.6,3) {L1};
\node(R2d) [BLANK] at (3.2,2) {L21};
\node(R2e) [BLANK] at (3.8,1) {L22};

\draw (R2a) -- (R2b);
\draw (R2b) -- (R2c);
\draw (R2c) -- (R2d);
\draw (R2d) -- (R2e);

\node(R3a) [RED] at (1.4,-1) {L12};
\node(R3b) [BLANK] at (2,0) {L11};
\node(R3c) [BLANK] at (2.6,1) {L1};
\node(R3d) [BLANK] at (3.2,0) {L21};
\node(R3e) [BLANK] at (3.8,-1) {L22};

\draw (R3a) -- (R3b);
\draw (R3b) -- (R3c);
\draw (R3c) -- (R3d);
\draw (R3d) -- (R3e);

\node(R4a) [BLANK] at (1.4,-3) {L12};
\node(R4b) [BLANK] at (2,-2) {L11};
\node(R4c) [BLANK] at (2.6,-1) {L1};
\node(R4d) [RED] at (3.2,-2) {L21};
\node(R4e) [BLANK] at (3.8,-3) {L22};

\draw (R4a) -- (R4b);
\draw (R4b) -- (R4c);
\draw (R4c) -- (R4d);
\draw (R4d) -- (R4e);

\node(R5a) [BLANK] at (1.4,-5) {L12};
\node(R5b) [BLANK] at (2,-4) {L11};
\node(R5c) [BLANK] at (2.6,-3) {L1};
\node(R5d) [BLANK] at (3.2,-4) {L21};
\node(R5e) [RED] at (3.8,-5) {L22};

\draw (R5a) -- (R5b);
\draw (R5b) -- (R5c);
\draw (R5c) -- (R5d);
\draw (R5d) -- (R5e);

\node(RR1a) [RED] at (5,3.5) {L1};
\node(RR1b) [BLANK] at (5.6,2.5) {L2};
\node(RR1c) [BLANK] at (6.2,1.5) {L3};

\draw (RR1a) -- (RR1b);
\draw (RR1b) -- (RR1c);

\node(RR2a) [BLANK] at (5,1) {L1};
\node(RR2b) [RED] at (5.6,0) {L2};
\node(RR2c) [BLANK] at (6.2,-1) {L3};

\draw (RR2a) -- (RR2b);
\draw (RR2b) -- (RR2c);

\node(RR3a) [BLANK] at (5,-1.5) {L1};
\node(RR3b) [BLANK] at (5.6,-2.5) {L2};
\node(RR3c) [RED] at (6.2,-3.5) {L3};

\draw (RR3a) -- (RR3b);
\draw (RR3b) -- (RR3c);

\end{tikzpicture}
\caption{Configuration after round $1$ of the \ms game, with a particular choice of Spoiler's round $1$ moves in red on the left, and Duplicator's oblivious responses in red on the right.}
\label{fig:after-first-round}
\end{figure}

We now show that Spoiler's second-round moves must be on the right side. We can focus entirely on the linear orders on the two sides. It suffices to consider Spoiler playing B3 in round $1$ (as B2 would be symmetric). Figure \ref{fig:orders-after-first-round}
depicts the linear orders on the two sides after round $1$.

\tikzset{RED/.style={draw=red, very thick, fill=none, font=\tiny, circle, inner sep=0.5mm}}
\tikzset{BLUE/.style={draw=blue, very thick, fill=none, font=\tiny, circle, inner sep=0.5mm}}
\tikzset{BLUESMALL/.style={draw=blue, very thick, fill=none, font=\tiny, circle, inner sep=0.3mm}}
\tikzset{BLANK/.style={draw=black, fill=none, font=\tiny, circle, inner sep=0.5mm}}
\tikzset{LABEL/.style={draw=none, fill=none, font=\small, circle, inner sep=0cm}}
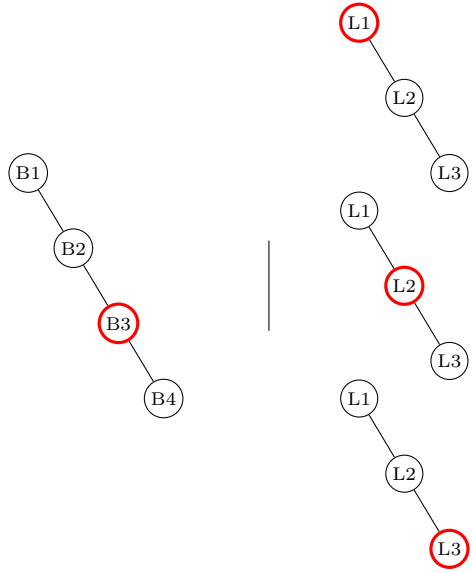
\begin{figure}[ht]
\centering
\begin{tikzpicture}

\node(L2a) [BLANK] at (-3,1.5) {B1};
\node(L2b) [BLANK] at (-2.4,0.5) {B2};
\node(L2c) [RED] at (-1.8,-0.5) {B3};
\node(L2d) [BLANK] at (-1.2,-1.5) {B4};
\draw (L2a) -- (L2b);
\draw (L2b) -- (L2c);
\draw (L2c) -- (L2d);

\draw (0.2, 0.6) -- (0.2, -0.6);

\node(RR1a) [RED] at (1.4,3.5) {L1};
\node(RR1b) [BLANK] at (2,2.5) {L2};
\node(RR1c) [BLANK] at (2.6,1.5) {L3};

\draw (RR1a) -- (RR1b);
\draw (RR1b) -- (RR1c);

\node(RR2a) [BLANK] at (1.4,1) {L1};
\node(RR2b) [RED] at (2,0) {L2};
\node(RR2c) [BLANK] at (2.6,-1) {L3};

\draw (RR2a) -- (RR2b);
\draw (RR2b) -- (RR2c);

\node(RR3a) [BLANK] at (1.4,-1.5) {L1};
\node(RR3b) [BLANK] at (2,-2.5) {L2};
\node(RR3c) [RED] at (2.6,-3.5) {L3};

\draw (RR3a) -- (RR3b);
\draw (RR3b) -- (RR3c);

\end{tikzpicture}
\caption{Linear orders with moves in red after the first round of the \ms game.}
\label{fig:orders-after-first-round}
\end{figure}

If Spoiler's second-round move on $\lo(4)$ is on B3, he already plays on top, and we are done.
If  Spoiler's second-round move is on B1 in $\lo(4)$, then Duplicator would play L1 in the second and third copies, and survive another round on one of these pairs. If Spoiler's second-round move is on B2, Duplicator would play L1 in the second copy and L2 in the third copy, and survive another round on one of these pairs. Finally, if Spoiler's second-round move is on B4, then Duplicator would play L3 on the second copy and survive one more round on this pair. It follows that Spoiler's round $2$ move must be $\playright$.
\item If Spoiler's third-round move is also on the right side, then it is easy to check that Duplicator can maintain a matching pair with the $\lo(4)$ on the left and the third copy of $\lo(3)$ shown on the right. This proves Claim 1. \qedhere
\end{enumerate}
\end{proof}

Now that Claim 1 has been established, we are ready to show that Spoiler cannot win the $3$-round MS game on $(\cA,\cB)$ without playing on top. We can make this argument by considering just the linear orders. By symmetry, we may assume that Spoiler's round $1$ move on the $\lo(4)$ is on B3. Let us focus on the second copy of $\lo(3)$ on the right. If Spoiler's second-round move on that copy is on L3, then Duplicator responds on B4 in $\lo(4)$ and wins on just this pair easily. Therefore, if Spoiler is to not move on top, he must play on L1 on this second copy of $\lo(3)$, and on L1 or L2 on the third copy of $\lo(3)$. If Spoiler plays on L2, then Duplicator responds on B2 in $\lo(4)$ and is guaranteed a matching pair between $\lo(4)$ and one of the copies of $\lo(3)$; if Spoiler plays on L1, then Duplicator responds on B1, and the same conclusion holds. 
 
This completes the proof that Spoiler must play on top in order to win the $3$-round \ms game on $(\cA, \cB)$. \end{proof}

We now analyze the variant of the \ms game in which Spoiler never plays on top. Recall that a \emph{type} over an $r$-tuple of distinct variables, $x_1, \ldots, x_r$, is a conjunction of atomic and negated atomic formulas such that for each atomic formula with variables from $x_1, \ldots, x_r$, exactly one of the atomic formula and its negation appears as a conjunct.

We assume that every {\fo}-sentence with $r$  quantifiers is in prenex normal form, that its quantifier-free part is a disjunction of types, and that the quantifier prefix is $Q_1x_1 \ldots Q_rx_r$, where each $Q_j$ is the quantifier $\exists$ or the quantifier  $\forall$. In other words, we assume that we work with {\fo}-sentences of the form $Q_1x_1 \ldots Q_rx_r \theta$, where $\theta$ is a disjunction of types.

\begin{defi} \label{defn:non-repl}
Let $\psi$ be a {\fo}-sentence of the form $Q_1x_1 \ldots Q_rx_r \theta$, where $\theta$ is a disjunction  of types.
\begin{itemize}
\item A type $t(x_1,\ldots,x_r)$ occurring as a disjunct of $\theta$ is \emph{non-replicating in $\psi$} if conditions $(a)$ and $(b)$ below  hold:
\begin{enumerate}
\item For every two distinct variables $x_i$ and $x_j$, if the equality $x_i = x_j$ appears in  $t(x_1,\ldots,x_r)$, then exactly one of the following two conditions holds:
\begin{enumerate}[label=(\roman*)]
\item Both variables $x_i$ and $x_j$ are universally quantified;
\item One of the two variables is universally quantified, the other variable is existentially quantified, and the existential quantifier appears before the universal quantifier in the quantifier prefix (that is, if $Q_i=\exists$ and $Q_j = \forall$, then we must have $i < j$).
\end{enumerate}
\item If an equality $c_i=x_j$ appears in $t(x_1,\ldots,x_r)$, then the variable $x_j$ is universally quantified.
\end{enumerate}
\item A type $t(x_1,\ldots,x_r)$ occurring as a disjunct of $\theta$ is \emph{replicating in $\psi$} if $t(x_1,\ldots,x_r)$ is not non-replicating in $\psi$.
\item $\psi$ is a \emph{non-replicating} sentence if every
type occurring as a disjunct of $\theta(x_1,\ldots,x_r)$ of $\psi$ is non-replicating in $\psi$.
\item $\psi$ is a \emph{replicating} sentence if it is not a non-replicating sentence (i.e.,~at least one type occurring as a disjunct of  $\theta(x_1,\ldots,x_r)$ of $\psi$ is replicating in $\psi$).
\end{itemize}
\end{defi}

We illustrate the notion of a non-replicating sentence with several examples in which the $<$ relation is assumed to range over the elements of rooted trees (with not all elements necessarily related); thus, instead of spelling out full types explicitly, we will often give a part of the type that determines the full type. 

\begin{exa}\label{exam:eq-free-non-repl}
If $\psi \equiv Q_1x_1 \ldots Q_rx_r \theta$ is such that every type  in $\theta$ is equality-free, then $\psi$ is non-replicating.
\end{exa}

\begin{exa}\label{exam:both-non-repl}
Consider the {\fo}-sentence
\begin{equation}
    \Phi_{3,\forall}: \forall x_1 \exists x_2 \exists x_3(x_1 < x_2 < x_3 \lor x_2 < x_3 < x_1),
\end{equation}
encountered at the beginning of this section. We claim that both
$ \Phi_{3,\forall}$ and  $\lnot  \Phi_{3,\forall}$ are non-replicating sentences.

To see this, first observe that the quantifier-free part of $ \Phi_{3,\forall}$ consists of two equality-free types, hence $ \Phi_{3,\forall}$ is a non-replicating sentence by Example \ref{exam:eq-free-non-repl}.

Secondly, the negation $\lnot\Phi_{3,\forall}$ of $\Phi_{3,\forall}$ is the {\fo}-sentence
\begin{equation}
\exists  x_1 \forall x_2 \forall  x_3 \lnot (x_1 < x_2 < x_3 \lor x_2 < x_3 < x_1).
\end{equation}
This is a non-replicating sentence because the possible equalities in the types of its quantifier-free part are 
$x_1=x_2$, $x_1=x_3$, and $x_2=x_3$. The first two  involve the existentially quantified variable $x_1$ and the universally quantified variables $x_2$ or $x_3$ (hence, they satisfy condition (ii) in Definition \ref{defn:non-repl}(a)), while the third  involves the two universally quantified variables $x_2$ and $x_3$ (hence, it satisfies condition (i) in Definition \ref{defn:non-repl}(a)).
\end{exa}

\begin{exa} \label{exam:one-non-repl}
Consider the {\fo}-sentence
\begin{eqnarray}
    \Phi_{3,\exists}: \exists x_1 \forall x_2 \exists x_3&\Big(&x_2<x_1 \rightarrow x_3 > x_1 ~ \land \notag\\
    &&x_2 >x_1 \rightarrow (x_3 \neq x_2 \land x_3 > x_1) ~ \land \notag\\
    &&x_2=x_1 \rightarrow x_3<x_1 ~~\Big) \label{phi3E}.
\end{eqnarray}
This sentence asserts that there is an element with one smaller element and two larger elements. We claim that 
$ \Phi_{3,\exists}$ is a non-replicating sentence, but   its negation $\lnot  \Phi_{3,\exists}$ is   replicating.

First, it is easy to see that the quantifier-free part of $\Phi_{3,\exists}$
is logically equivalent to a disjunction of types in which 
the only equality is $x_1=x_2$. Since   $x_1$ is  existentially quantified and $x_2$ is universally quantified in $\Phi_{3,\exists}$, the sentence $\Phi_{3,\exists}$ is non-replicating.
The negation $\lnot  \Phi_{3,\exists}$ of $ \Phi_{3,\exists}$ is the {\fo}-sentence
\begin{eqnarray}
     \forall  x_1 \exists x_2 \forall  x_3& \lnot \Big (&x_2< x_1 \rightarrow x_3 > x_1~\land \notag\\
    &&x_2 >x_1 \rightarrow (x_3 \neq x_2 \land x_3 > x_1)~\land \notag\\
    &&x_2=x_1 \rightarrow x_3<x_1 ~~\Big),
\end{eqnarray}
By pushing the negation inside, it is easy to see that, as a disjunction of types, the quantifier-free part of $\lnot\Phi_{3,\exists}$ includes the type $x_1=x_2 < x_3$ as a disjunct. Since  $x_1$  is universally quantified  and $x_2$ is existentially  quantified   in $\neg \Phi_{3,\exists}$, the sentence $\lnot  \Phi_{3,\exists}$ is  replicating.
\end{exa}

We are now ready to state and prove the main result about the expressive power of the variant of the \ms game in which  Spoiler never plays on top.

\begin{thm} \label{thm:non-repl}
 Let $r \in \N$, and let $\cA$ and $\cB$  be two sets of $\tau$-structures. Then the following statements are equivalent:
\begin{enumerate}
\item Spoiler wins the $r$-round \ms game on $(\cA,\cB)$ without ever playing on top.
\item There is a separating sentence $\psi$ for $(\cA,\cB)$  of the form $Q_1x_1\ldots Q_rx_r\theta(x_1, \ldots, x_r)$, where $\theta$ is quantifier-free; moreover, if $\cS$ is a winning strategy for Spoiler 
obtained from $\psi$, then for every pebbled structure
$\langle\bA ~|~ a_1,\ldots,a_r\rangle$ with $\bA \in \cA$, and for every pebbled structure $\langle\bB ~|~ b_1,\ldots,b_r\rangle$ with $\bB \in \cB$ arising by using this strategy, the following hold:
\begin{enumerate}[label=(\roman*)]
\item The  disjunct  of $\theta$ satisfied by $(a_1,\ldots,a_r)$ is a non-replicating type in $\psi$.
\item The  disjunct  of $\lnot \theta$  satisfied by $(b_1,\ldots,b_r)$ is a non-replicating type in $\lnot \psi$.
\end{enumerate}
\end{enumerate}
\end{thm}
\begin{proof}
Assume  that  Spoiler wins the $r$-round \ms game on $\cA$ and $\cB$ without ever playing on top. By Theorem \ref{thm:m-s-game-strong}, there is a {\fo}-sentence 
$\psi \equiv Q_1x_1 \ldots Q_rx_r \theta(x_1,\ldots,x_r)$ with $r$ quantifiers such that $\psi$ is separating for  $(\cA,\cB)$; moreover, if $\cS$
is a winning strategy for Spoiler obtained from $\psi$, then for every pebbled structure
$\langle\bA ~|~ a_1,\ldots,a_r\rangle$ with $\bA \in \cA$, and every pebbled structure $\langle\bB ~|~ b_1,\ldots,b_r\rangle$ with $\bB \in \cB$ arising by using this strategy, we have that $\bA \models \theta(a_1,\ldots,a_r)$ and $\bB \models \lnot \theta(b_1,\ldots,b_r)$. Fix two such pebbled structures
$\langle\bA ~|~ a_1,\ldots,a_r\rangle$ with $\bA \in \cA$ and $\langle\bB ~|~ b_1,\ldots,b_r\rangle$ with $\bB \in \cB$. Let $t(x_1,\ldots,x_r)$ be the type that occurs as a disjunct of $\theta(x_1,\ldots,x_r)$
and is satisfied by $(a_1,\ldots,a_r)$. We claim that $t(x_1,\ldots,x_r)$ is a non-replicating type in $\psi$.  Otherwise, one of the following three things would happen: 
(1) $t(x_1,\ldots,x_r)$ would contain an equality $x_i=x_j$, where both the $i$-th and the $j$-th quantifier of $\psi$ are $\exists$; this implies that $a_i=a_j$, but since both $a_i$ and $a_j$ were played by  Spoiler, this means that Spoiler  played on top, which is a contradiction.
(2) $t(x_1,\ldots,x_r)$ would contain an equality $x_i=x_j$, where the $i$-th quantifier of $\psi$ is $\forall$, the $j$-th quantifier of $\psi$ is $\exists$, and $i<j$; this implies that $a_i=a_j$, but since $a_j$ was played by  Spoiler, this means that Spoiler  played on top, which is a contradiction. (3)  $t(x_1,\ldots,x_r)$ would contain an equality $c_i=x_j$, where the $j$-th quantifier of $\psi$  is $\exists$; this implies that $c_i=a_j$, but $a_j$ was played by Spoiler, hence Spoiler  played on top, which is a contradiction.

A similar argument shows that the disjunct of $\lnot \theta$ satisfied by $(b_1,\ldots,b_r)$ is a non-replicating type in $\lnot \psi$.

Conversely, assume  that there is a separating sentence $\psi$ for $(\cA,\cB)$
of the form
$ Q_1x_1\ldots Q_rx_r\theta(x_1, \ldots, x_r)$, where $\theta$ is quantifier-free; further, assume that if $\cS$ is a winning strategy for Spoiler obtained from $\psi$, then  the properties asserted in (b) hold. We claim that  Spoiler never plays on top using  strategy $\cS$. Consider a pebbled structure $\langle\bA ~|~ a_1,\ldots,a_r\rangle$ with $\bA \in \cA$ and a pebbled structure $\langle\bB ~|~ b_1,\ldots,b_r\rangle$ with $\bB \in \cB$ arising by using this strategy.
Let us examine the $j$-th move of Spoiler, where $1\leq j\leq r$. There are two cases to consider, namely, the case in which $Q_j=\exists$ and the case in which $Q_j=\forall$. If $Q_j=\exists$, then Spoiler has to place the $j$-th pebble on an element of $\bA$.
Since  Spoiler plays using the strategy ${\cS}$, Part (a) of Theorem \ref{thm:m-s-game-strong} tells that $\bA \models \theta(a_1,\ldots,a_r)$.
Consequently, the tuple $(a_1,\ldots,a_r)$ must satisfy a type
$t(x_1,\ldots,x_r)$ occurring as a disjunct of $\theta$.
Thus,  $t(x_1,\ldots,x_r)$ must be a non-replicating type in $\psi$.  Since the $j$-th quantifier of $\psi$ is $\exists$, the type  $t(x_1,\ldots,x_r)$  cannot contain an equality of the form $x_i = x_j$ with $i<j$ or an equality of the form $c_i=x_j$; instead, it must contain $\lnot (x_i=x_j)$ and $\lnot (c_i=x_j)$. Therefore, $a_i\not = a_j$, for every $i<j$, and also $c_i \not = a_j$, where $c_i$ is a constant; hence,  Spoiler does not play on top on $\bA$. If $Q_j=\forall$, the argument is similar using the fact that the type in $\lnot\theta$ satisfied by $(b_1,\ldots,b_r)$ must be non-replicating in $\lnot\psi$.
\end{proof}

\begin{cor}\label{cor-non-repl}
Let $r\in\N$, and let $\cA$ and $\cB$ be two sets of $\tau$-structures. If there is a separating sentence $\psi$ for $(\cA,\cB)$ with $r$ quantifiers such that both $\psi$ and $\lnot \psi$ are non-replicating sentences, then  Spoiler wins the $r$-round \ms game on $(\cA,\cB)$ without ever playing on top. 
\end{cor}

Consider the variant of the \ms game in which  Spoiler never plays on top on the left and the variant of the \ms game in which  Spoiler never plays on top on the right. The proof of Theorem \ref{thm:non-repl} can be adapted to yield analogous results for these two variants. For example, we have the following results, whose proofs we omit.

\begin{thm} \label{thm:left-non-repl}
 Let $r \in \N$, and let $\cA$ and $\cB$ be two sets of $\tau$-structures. Then the following statements are equivalent:
\begin{enumerate}
\item Spoiler wins the $r$-round \ms game on $(\cA,\cB)$ without ever playing on top on the left side.
\item There is a separating sentence $\psi$ for $(\cA,\cB)$  of the form $Q_1x_1\ldots Q_rx_r\theta(x_1, \ldots, x_r)$, where $\theta$ is quantifier-free; moreover, if $\cS$ is a winning strategy for Spoiler obtained from $\psi$, then for every $\langle\bA ~|~ a_1,\ldots,a_r\rangle$ with $\bA \in \cA$ arising by using this strategy, the  disjunct of $\theta$ satisfied by $(a_1,\ldots,a_r)$ is a non-replicating type in $\psi$.
\end{enumerate}
\end{thm}

\begin{cor}\label{cor:left}
Let $r \in \N$, and let $\cA$ and $\cB$ be two sets of $\tau$-structures. If there is a separating sentence $\psi$ for $(\cA,\cB)$
 with $r$ quantifiers such that  $\psi$ is a  non-replicating sentence, then  Spoiler wins the $r$-round \ms game on $(\cA,\cB)$ without ever playing on top on the left side.
\end{cor}

We revisit Example \ref{exam:both-non-repl}. Since the sentence $\Phi_{3,\forall}$ asserts that every element has two smaller elements or two larger elements,  $\Phi_{3,\forall}$ is a separating sentence for $(\{\lo(4)\},\{\lo(3)\})$. According to Example \ref{exam:both-non-repl}, both
$\Phi_{3,\forall}$ and $\lnot \Phi_{3,\forall}$ are non-replicating sentences, hence Corollary \ref{cor-non-repl} implies that Spoiler can win the $3$-round \ms game on $(\{\lo(4)\}, \{\lo(3)\})$ without playing on top. As mentioned earlier, we can also verify this directly.
Note that $\Phi_{3,\forall}$ is  a separating sentence for
$(\{\lo(4)\},\{\rt(3),\lo(3)\})$ as well, hence Corollary \ref{cor-non-repl} implies that Spoiler can win the $3$-round \ms game on $(\{\lo(4)\}, \{\rt(3),\lo(3)\})$ without playing on top.

 For a different application, we claim that Spoiler can win the $3$-round \ms game on $(\{\rt(4),\lo(4)\},\{\lo(3)\})$ without playing on top. For this,
we first let  $\mathrm{inc} (x,y)$ be the formula $\neg(x<y) \land \neg(y<x) \land (x\neq y)$, asserting that $x$ and $y$ are distinct incomparable elements, and then let $\Psi_{3,\forall}$ be the sentence
\begin{eqnarray*}
     \Psi_{3,\forall}:\forall x_1 \exists x_2 \exists x_3(x_1 < x_2 < x_3 \lor x_2 < x_3 < x_1 \lor 
    \mathrm{inc}(x_1,x_2)\land \mathrm{inc}(x_1,x_3)\land x_2\not = x_3),
\end{eqnarray*}
which asserts that every element has two bigger elements, or two smaller elements, or two different elements that are incomparable with it. Clearly,
$\Psi_{3,\forall}$ is a separating sentence for $(\{\rt(4),\lo(4)\},\{\lo(3)\})$; moreover, it is easy to see that both $\Psi_{3,\forall}$ and its negation are non-replicating sentences, hence Corollary  \ref{cor-non-repl} implies that Spoiler can win the $3$-round \ms game on $(\{\rt(4),\lo(4)\}, \{\lo(3)\})$ without playing on top.

Finally, Proposition \ref{prop:rooted} and Corollary \ref{cor-non-repl} imply that there is no {\fo}-sentence $\psi$ with $3$ quantifiers such that $\psi$ is a separating sentence for $(\{\rt(4),\lo(4)\},\{\rt(3),\lo(3)\})$ and both $\psi$ and $\neg \psi$ are non-replicating sentences.

Corollary \ref{cor-non-repl} gives a sufficient condition for Spoiler to win  the $r$-round \ms~game without ever playing on top, namely, it suffices to have a
separating sentence such that both the sentence and its negation are non-replicating. It is natural to ask whether this condition is also necessary, in which case it would replace the more complicated necessary and sufficient condition in Theorem \ref{thm:left-non-repl}.
The next result answers this question in the  negative.

\tikzset{RED/.style={draw=red, very thick, fill=none, font=\tiny, circle, inner sep=0.5mm}}
\tikzset{BLUE/.style={draw=blue, very thick, fill=none, font=\tiny, circle, inner sep=0.5mm}}
\tikzset{BLUESMALL/.style={draw=blue, very thick, fill=none, font=\tiny, circle, inner sep=0.3mm}}
\tikzset{BLANK/.style={draw=black, fill=none, font=\tiny, circle, inner sep=0.5mm}}
\tikzset{LABEL/.style={draw=none, fill=none, font=\small, circle, inner sep=0cm}}
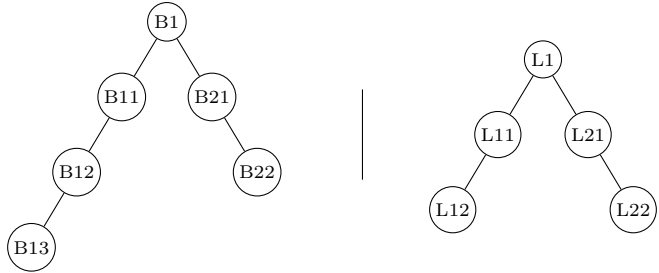
\begin{figure}[ht]
\centering
\begin{tikzpicture}

\node(L1a) [BLANK] at (-4.2,-1.5) {B13};
\node(L1b) [BLANK] at (-3.6,-0.5) {B12};
\node(L1c) [BLANK] at (-3,0.5) {B11};
\node(L1d) [BLANK] at (-2.4,1.5) {B1};
\node(L1e) [BLANK] at (-1.8,0.5) {B21};
\node(L1f) [BLANK] at (-1.2,-0.5) {B22};
\draw (L1a) -- (L1b);
\draw (L1b) -- (L1c);
\draw (L1c) -- (L1d);
\draw (L1d) -- (L1e);
\draw (L1e) -- (L1f);

\draw (0.2, 0.6) -- (0.2, -0.6);

\node(R1a) [BLANK] at (1.4,-1) {L12};
\node(R1b) [BLANK] at (2,0) {L11};
\node(R1c) [BLANK] at (2.6,1) {L1};
\node(R1d) [BLANK] at (3.2,0) {L21};
\node(R1e) [BLANK] at (3.8,-1) {L22};

\draw (R1a) -- (R1b);
\draw (R1b) -- (R1c);
\draw (R1c) -- (R1d);
\draw (R1d) -- (R1e);

\end{tikzpicture}
\caption{A rooted tree $\rt(4)$ whose longest branch has $4$ nodes (left); a rooted tree $\rt(3)$ whose longest branch has $3$ nodes (right).}
\label{fig:2_trees}
\end{figure}

\begin{thm}\label{thm:repl-v-sep}
Let $\rt(4)$ and $\rt(3)$ be the rooted trees in Figure \ref{fig:2_trees}. Then the following statements are true.
\begin{enumerate}[label=(\alph*)]
\item Spoiler can win the $3$-round \ms game on $(\{\rt(4)\},\{\rt(3)\})$ without playing on top.
\item There is no {\fo}-sentence $\psi$ with $3$ quantifiers such that 
\begin{enumerate}[label=(\roman*)]
\item $\psi$ is a separating sentence for $(\{\rt(4)\},\{\rt(3)\})$;
\item both $\psi$ and $\lnot\psi$ are non-replicating sentences.
\end{enumerate}
\end{enumerate}
\end{thm}
\begin{proof}
We claim that Spoiler can win the $3$-round \ms game on $(\{\rt(4)\},\{\rt(3)\})$ without playing on top by playing his first-round moves on the left, his second-round moves on the right, and his third-round moves on the left as described below.
\begin{itemize}
\item Spoiler plays first by placing a pebble on B12 in $\rt(4)$ on the left.
\item After Duplicator makes five copies of $\rt(3)$ on the right and places a pebble on a different element of each copy, Spoiler places pebbles on each copy of $\rt(3)$ as follows:
\begin{enumerate}
\item In response to L1 on the first copy of $\rt(3)$, Spoiler pebbles L11.

\item  In response to L11 on the second copy of $\rt(3)$, Spoiler pebbles L1, 

\item In response to L12 on the third copy of $\rt(3)$, Spoiler pebbles L21.
\item In response to L21 on the fourth copy of $\rt(3)$, Spoiler pebbles L1.
\item In response to L22 on the fifth copy of $\rt(3)$, Spoiler pebbles L11.

 \end{enumerate}
\item After Duplicator makes six copies of $\rt(4)$ on the left and places a pebble on a different element of each copy, Spoiler places pebbles on each copy of $\rt(4)$ as indicated in Figure \ref{fig:game_tree}, where the first-round moves of the two players are colored red, the second-round-moves of the two players are colored blue, and Spoiler's third-round moves are colored green.

\tikzset{RED/.style={draw=red, very thick, fill=none, font=\small, circle, inner sep=1mm}}
\tikzset{BLUE/.style={draw=blue, very thick, fill=none, font=\small, circle, inner sep=1mm}}
\tikzset{GREEN/.style={draw=black!40!green, very thick, fill=none, font=\small, circle, inner sep=1mm}}
\tikzset{BLUESMALL/.style={draw=blue, very thick, fill=none, font=\small, circle, inner sep=0.6mm}}
\tikzset{BLANK/.style={draw=black, fill=none, font=\small, circle, inner sep=1mm}}
\tikzset{LABEL/.style={draw=none, fill=none, font=\small, circle, inner sep=0cm}}
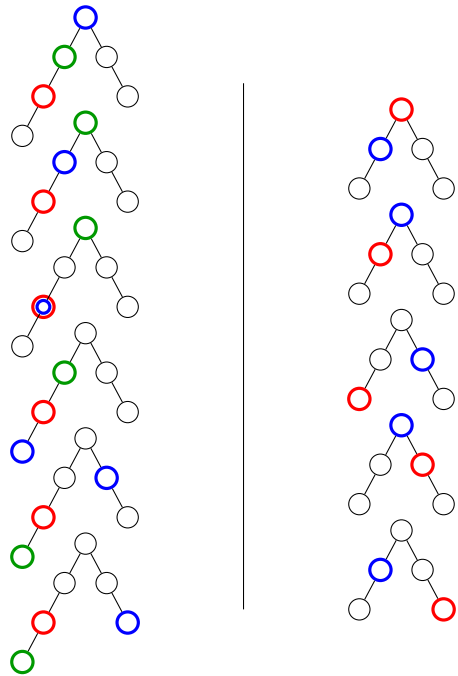
\begin{figure}[ht]
\centering
\begin{tikzpicture}[scale=.7]

\node(L1a) [BLANK] at (-4.2,6) {};
\node(L1b) [RED] at (-3.8,6.75) {};
\node(L1c) [GREEN] at (-3.4,7.5) {};
\node(L1d) [BLUE] at (-3,8.25) {};
\node(L1e) [BLANK] at (-2.6,7.5) {};
\node(L1f) [BLANK] at (-2.2,6.75) {};
\draw (L1a) -- (L1b);
\draw (L1b) -- (L1c);
\draw (L1c) -- (L1d);
\draw (L1d) -- (L1e);
\draw (L1e) -- (L1f);

\node(L2a) [BLANK] at (-4.2,4) {};
\node(L2b) [RED] at (-3.8,4.75) {};
\node(L2c) [BLUE] at (-3.4,5.5) {};
\node(L2d) [GREEN] at (-3,6.25) {};
\node(L2e) [BLANK] at (-2.6,5.5) {};
\node(L2f) [BLANK] at (-2.2,4.75) {};
\draw (L2a) -- (L2b);
\draw (L2b) -- (L2c);
\draw (L2c) -- (L2d);
\draw (L2d) -- (L2e);
\draw (L2e) -- (L2f);

\node(L3a) [BLANK] at (-4.2,2) {};
\node(L3bb) [RED] at (-3.8,2.75) {};
\node(L3b) [BLUESMALL] at (-3.8,2.75) {};
\node(L3c) [BLANK] at (-3.4,3.5) {};
\node(L3d) [GREEN] at (-3,4.25) {};
\node(L3e) [BLANK] at (-2.6,3.5) {};
\node(L3f) [BLANK] at (-2.2,2.75) {};
\draw (L3a) -- (L3b);
\draw (L3b) -- (L3c);
\draw (L3c) -- (L3d);
\draw (L3d) -- (L3e);
\draw (L3e) -- (L3f);

\node(L4a) [BLUE] at (-4.2,0) {};
\node(L4b) [RED] at (-3.8,0.75) {};
\node(L4c) [GREEN] at (-3.4,1.5) {};
\node(L4d) [BLANK] at (-3,2.25) {};
\node(L4e) [BLANK] at (-2.6,1.5) {};
\node(L4f) [BLANK] at (-2.2,0.75) {};
\draw (L4a) -- (L4b);
\draw (L4b) -- (L4c);
\draw (L4c) -- (L4d);
\draw (L4d) -- (L4e);
\draw (L4e) -- (L4f);

\node(L5a) [GREEN] at (-4.2,-2) {};
\node(L5b) [RED] at (-3.8,-1.25) {};
\node(L5c) [BLANK] at (-3.4,-0.5) {};
\node(L5d) [BLANK] at (-3,0.25) {};
\node(L5e) [BLUE] at (-2.6,-0.5) {};
\node(L5f) [BLANK] at (-2.2,-1.25) {};
\draw (L5a) -- (L5b);
\draw (L5b) -- (L5c);
\draw (L5c) -- (L5d);
\draw (L5d) -- (L5e);
\draw (L5e) -- (L5f);

\node(L6a) [GREEN] at (-4.2,-4) {};
\node(L6b) [RED] at (-3.8,-3.25) {};
\node(L6c) [BLANK] at (-3.4,-2.5) {};
\node(L6d) [BLANK] at (-3,-1.75) {};
\node(L6e) [BLANK] at (-2.6,-2.5) {};
\node(L6f) [BLUE] at (-2.2,-3.25) {};
\draw (L6a) -- (L6b);
\draw (L6b) -- (L6c);
\draw (L6c) -- (L6d);
\draw (L6d) -- (L6e);
\draw (L6e) -- (L6f);

\draw (0, 7) -- (0, -3);

\node(R1a) [BLANK] at (2.2,5) {};
\node(R1b) [BLUE] at (2.6,5.75) {};
\node(R1c) [RED] at (3,6.5) {};
\node(R1d) [BLANK] at (3.4,5.75) {};
\node(R1e) [BLANK] at (3.8,5) {};
\draw (R1a) -- (R1b);
\draw (R1b) -- (R1c);
\draw (R1c) -- (R1d);
\draw (R1d) -- (R1e);

\node(R2a) [BLANK] at (2.2,3) {};
\node(R2b) [RED] at (2.6,3.75) {};
\node(R2c) [BLUE] at (3,4.5) {};
\node(R2d) [BLANK] at (3.4,3.75) {};
\node(R2e) [BLANK] at (3.8,3) {};
\draw (R2a) -- (R2b);
\draw (R2b) -- (R2c);
\draw (R2c) -- (R2d);
\draw (R2d) -- (R2e);

\node(R3a) [RED] at (2.2,1) {};
\node(R3b) [BLANK] at (2.6,1.75) {};
\node(R3c) [BLANK] at (3,2.5) {};
\node(R3d) [BLUE] at (3.4,1.75) {};
\node(R3e) [BLANK] at (3.8,1) {};
\draw (R3a) -- (R3b);
\draw (R3b) -- (R3c);
\draw (R3c) -- (R3d);
\draw (R3d) -- (R3e);

\node(R4a) [BLANK] at (2.2,-1) {};
\node(R4b) [BLANK] at (2.6,-0.25) {};
\node(R4c) [BLUE] at (3,0.5) {};
\node(R4d) [RED] at (3.4,-0.25) {};
\node(R4e) [BLANK] at (3.8,-1) {};
\draw (R4a) -- (R4b);
\draw (R4b) -- (R4c);
\draw (R4c) -- (R4d);
\draw (R4d) -- (R4e);

\node(R5a) [BLANK] at (2.2,-3) {};
\node(R5b) [BLUE] at (2.6,-2.25) {};
\node(R5c) [BLANK] at (3,-1.5) {};
\node(R5d) [BLANK] at (3.4,-2.25) {};
\node(R5e) [RED] at (3.8,-3) {};
\draw (R5a) -- (R5b);
\draw (R5b) -- (R5c);
\draw (R5c) -- (R5d);
\draw (R5d) -- (R5e);

\end{tikzpicture}
\caption{The configuration of the $3$-round \ms game after Spoiler's moves in the $3$ rounds. Note that Spoiler has played $\pebbleleft$ in round 1 with red, $\pebbleright$ in round 2 with blue, and $\pebbleleft$ in round 3 with green.}
\label{fig:game_tree}
\end{figure}

\end{itemize}
From Figure \ref{fig:game_tree}, it is clear that Duplicator can not play on the right in the third round so that a partial isomorphism is maintained between one of the six copies of $\rt(4)$ and one of the five copies of $\rt(3)$. For instance, in four copies of $\rt(4)$, the green pebble played by Spoiler are above the red one; in order to maintain a partial isomorphism with one of these copies, there are only two copies of $\rt(3)$ for Duplicator as candidates -- and in both of them, any choice of a green pebble above a red pebble would break all isomorphisms with the copies of $\rt(4)$ on the left. Similarly, there are two copies of $\rt(4)$ with a green pebble below a red pebble, so Duplicator has only two candidates among the copies of $\rt(3)$ to place a green pebble below a red pebble -- and in both of them, the blue pebble is comparable with at least one other pebble, which is not true in the corresponding copies of $\rt(4)$ on the left. Thus, Spoiler wins the $3$-round \ms game on 
$(\{\rt(4)\},\{\rt(3)\})$.
Moreover, only B12 in the third copy of $\rt(4)$ has two pebbles placed on it with the red pebble placed by Spoiler in the first round and the blue pebble placed by Duplicator in the second round. 
Thus,  Spoiler wins the $3$-round \ms game on 
$(\{\rt(4)\},\{\rt(3)\})$ without playing on top. This completes the proof of the first part of the theorem.

To prove the second part of the theorem, we claim that Spoiler wins the $3$-round \ms game on $(\{\rt(4)\},\{\rt(3)\})$ without playing on top only if Spoiler plays on the left in the first round, on the right in the second round, and on the left in the third round; moreover, Spoiler's round first move on $\rt(4)$ must be on either B11 or B12. The argument about Spoiler's first-round move is similar to the one in the proof of Proposition \ref{prop:rooted}.

We  claim that if Spoiler's first-round move is on the right, then Duplicator can maintain a partial isomorphism between one copy of $\rt(4)$ and one copy of $\rt(3)$ for three rounds. Indeed, in response to a first-round move by Spoiler on L1 in $\rt(3)$,  Duplicator plays on B1 in one copy of $\rt(4)$, and then easily survives two more rounds (e.g.,~in response to a second-round Spoiler play on B12, Duplicator plays on L11 in one copy and L12 in another copy of the $\rt(3)$). First-round Spoiler plays on L11 or L21 are met by a Duplicator play on B21 in one copy of $\rt(4)$, and first-round Spoiler plays on L12 or L22 are met by a Duplicator play on B22 in a copy of $\rt(4)$; in each case, Duplicator readily survives two more rounds on just this pair of structures. 
If Spoiler's first-round move on $\rt(4)$ is not on B11 or on B12, then once again, Duplicator survives three rounds just on  $(\{\rt(4)\}, \{\rt(3)\})$ by maintaining a partial isomorphism between two copies via the following responses in a copy of $\rt(3)$: B1$\hookrightarrow$L1, B13$\hookrightarrow$L12, B21$\hookrightarrow$L21, B22$\hookrightarrow$L22.

We have established that Spoiler wins only if his first-round move is on B12 or on B12 on the left. Using this, it is easy to verify that Spoiler must play his second-round moves on the right and his third-round moves on the left. Otherwise, Duplicator easily survives two more rounds. For example, if Spoiler plays B12 in the first round and B1 in the second round, then Duplicator plays L12 and L1 in the third copy of $\rt(3)$, and then survives the third round.

So far, we have shown that  Spoiler wins the $3$-round \ms game on $(\{\rt(4)\},\{\rt(3)\})$ without playing on top only if  Spoiler plays on the left in the first round, on the right in the second round, and on the left in the third round.
Suppose now that there is a {\fo}-sentence $\psi$ with $3$ quantifiers that is separating for $(\{\rt(4)\},\{\rt(3)\})$ and such that both $\psi$ and $\neg \psi$ are non-replicating. Corollary \ref{cor-non-repl} implies that if Spoiler follows $\psi$, then Spoiler wins the $3$-move \ms game on $(\{\rt(4)\},\{\rt(3)\})$ without playing on top. 
Therefore, $\psi$ must have a quantifier prefix of the form $\exists\forall\exists$. Proposition \ref{prop:3alt}, however, implies that at least one of the sentences $\psi$ and $\neg \psi$ must be a non-replicating sentence, which is a contradiction. The proof of the theorem is now complete.
\end{proof}

The proof of Theorem \ref{thm:repl-v-sep} shows that
there are situations in which for Spoiler to win without ever playing on top, he must alternate between the left side and the right side more than once. In such cases, the separating sentence extracted from Spoiler's winning strategy has at least three alternations of quantifiers. The next proposition asserts that no sentence with three alternations of quantifiers has the property that both the sentence and its negation are non-replicating sentences.

\begin{prop}\label{prop:3alt}
Let $Q_1x_1\ldots Q_rx_r\theta$ be a {\fo}-sentence with $r$ variables such that $r \geq 3$ and $\theta$ is quantifier-free. If both $\psi$ and $\lnot\psi$ are non-replicating sentences, then neither $\psi$ nor $\lnot\psi$ contains the quantifiers $\exists x_i$, $\forall x_j$, $\exists x_k$  with $i<j<k$ in its prefix.
\end{prop}
\begin{proof}
Towards a contradiction and without loss of generality assume that $\psi$ and $\lnot\psi$ are non-replicating sentences such that the quantifier prefix of $\psi$
contains the quantifiers $\exists x_i$, $\forall x_j$, $\exists x_k$  with $i<j<k$. The quantifier-free parts $\theta$ of $\psi$ and $\lnot\theta$ of $\lnot\psi$ can be written as disjunctions of consistent types, so that each consistent type is a disjunct 
either of $\theta$ or of  $\lnot\theta$ (but not of both). Since $\psi$ is a non-replicating sentence, no type occurring
as a disjunct of $\theta$ can contain the equalities $x_i=x_k$ or $x_j=x_k$. Hence,
at least one of the types $t(x_1,\ldots,x_r)$ occurring as a disjunct of $\lnot\theta$ must contain the equalities $x_i=x_k$ and $x_j=x_k$. Since $t(x_1,\ldots,x_r)$ is a consistent type, it follows that $t(x_1,\ldots,x_r)$ must also contain the equality $x_i=x_j$. Since
$x_i$ is universally quantified in $\lnot\psi$ and $x_j$ is existentially quantified in $\lnot\psi$, it follows that $\lnot\psi$ is a replicating sentence, contradicting the assumption made.
\end{proof}

\section{Restricting the Number of Variables}\label{sec:finitepebbles}

Suppose in addition to the number of quantifiers, we also wish to simultaneously capture the number of \emph{variables} needed to express a certain property, or distinguish two sets of $\tau$-structures. How do we achieve this?

Let $\mathrm{FO}^k(\tau)$ be the set of well-formed {\fo} formulas over the schema $\tau$, that only use the
variables $\{x_1, \ldots, x_k\}$.

\subsection{Limitations of \ms Games}

It might at first seem reasonable to try to adapt the \ms game (which already captures the number of quantifiers) to also capture the number of variables, simply by limiting the number of pebble colors that can be used. This approach works in a straightforward manner with \ef games \cite[Definition 6.2, Theorem 6.10]{immermanbook}. There are two natural ways to define such adaptations of the \ms game.

\begin{game}[\ms Game with Repebbling]\label{msgame1}
Define the \emph{$r$-round, $k$-color \ms game with repebbling on $(\cA, \cB)$} as identical to the $r$-round \ms game on $(\cA, \cB)$, except that the set $\cC$ of pebble colors satisfies $|\cC| = k$, and so, Spoiler needs to play with at most $k$ pebble colors (forcing him to possibly re-use the same pebble color in two different rounds).
When a pebble of a given color is re-used, the previously played pebble of the same color is ``picked up'' from all structures. In every round, Duplicator has to respond with the same pebble color as the one used by Spoiler in that round. As before, the winning conditions are the same, i.e.,~at the end of each of $r$ rounds, Duplicator needs to exhibit a pebbled structure on the left and a pebbled structure on the right forming a matching pair, while Spoiler needs to exhibit a configuration within $r$ rounds where no pair of pebbled structures from the left and right form a matching pair.
\end{game}

While Game \ref{msgame1} seems like a reasonable candidate for capturing {\fo}-distinguishability with $r$ quantifiers and $k$ variables, this turns out to not work. Recall that the discussion following Figure \ref{fig:3_vs_2} observed that there is a $3$-quantifier $2$-variable sentence that is separating for $(\{\lo(3)\}, \{\lo(2)\})$, namely: $\exists x (\exists y (x < y) \land \exists y (y < x))$. So, if Game \ref{msgame1} were to capture distinguishability with $r$ quantifiers and $k$ variables, then Spoiler would need to have a winning strategy in the $3$-round, $2$-color \ms game with repebbling on this instance. The following lemma, whose proof is straightforward and omitted, shows that this is not the case.

\begin{lem} \label{weak_indistinguishabiliity_lemma}
Duplicator has a winning strategy in the $3$-round, $2$-color \ms game with repebbling on $(\{\lo(3)\},\{\lo(2)\})$.
\end{lem}

In order to define the second variant of the \ms game, we first need a definition.

\begin{defi}
During a complete play of the $r$-round, $k$-color \ms game with repebbling on $(\cA, \cB)$, we say that a pebbled structure $\langle\bS ~|~ s_1, \ldots, s_t\rangle$ is a \emph{parent} of a pebbled structure $\langle\bS ~|~ s'_1, \ldots, s'_{t'}\rangle$, if both of the following statements hold:
\begin{itemize}
\item the configuration containing $\langle\bS ~|~ s_1, \ldots, s_t\rangle$ is from the round immediately preceding the configuration containing $\langle\bS ~|~ s'_1, \ldots, s'_{t'}\rangle$,
\item $\langle\bS ~|~ s'_1, \ldots, s'_{t'}\rangle$ is the result of playing a new pebble color or reusing a pebble color on $\langle\bS ~|~ s_1, \ldots, s_t\rangle$.
\end{itemize}
Note that this means $t' = t$ (if a color was reused), or $t' = t + 1$ (if a new color was used).
\end{defi}

Observe that in the \ms game with repebbling, there may be $\bA, \bA' \in \cA$ and $\bB, \bB' \in \cB$ such that the following are both true for some $0 \leq t < r$:
\begin{itemize}
\item $\langle\bA' ~|~ a'_1, \ldots, a'_{t'}\rangle$ and $\langle\bB' ~|~ b'_1, \ldots, b'_{t'}\rangle$ are a matching pair in round $t+1$.
\item there are no matching pairs $\langle\bA ~|~ a_1, \ldots, a_t\rangle$ and $\langle\bB ~|~ b_1, \ldots, b_t\rangle$ in round $t$ such that $\langle\bA ~|~ a_1, \ldots, a_t\rangle$ is a parent of $\langle\bA' ~|~ a'_1, \ldots, a'_{t'}\rangle$, and $\langle\bB ~|~ b_1, \ldots, b_t\rangle$ is a parent of $\langle\bB' ~|~ b'_1, \ldots, b'_{t'}\rangle$.
\end{itemize}
Our second variant of the \ms game will constrain this heavily. First we need a definition.

\begin{defi}\label{def:hereditarymatch}
During a complete play of the $r$-round, $k$-color \ms game with repebbling on $(\cA, \cB)$, two $t$-pebbled structures $\langle\bA ~|~ a_1, \ldots, a_t\rangle$ on the left and $\langle\bB ~|~ b_1, \ldots, b_t\rangle$ on the right form a \emph{hereditary match} if the following hold:
\begin{enumerate}
\item $\langle\bA ~|~ a_1, \ldots, a_t\rangle$ and $\langle\bB ~|~ b_1, \ldots, b_t\rangle$ are within the same configuration, say at the end of round $t'$.
\item There is a sequence $\bA_0,\ldots,\bA_{t'}$ of pebbled structures with $\bA_0 = \bA$ and $\bA_{t'} = \langle\bA ~|~ a_1, \ldots, a_t\rangle$, such that $\bA_i$ is the parent of $\bA_{i+1}$ for each $0 \leq i < t'$.
\item There is a sequence $\bB_0,\ldots,\bB_{t'}$ of pebbled structures with $\bB_0 = \bB$ and $\bB_{t'} = \langle\bB ~|~ b_1, \ldots, b_t\rangle$, such that $\bB_i$ is the parent of $\bB_{i+1}$ for each $0 \leq i < t'$.
\item $\bA_i$ and $\bB_i$ form a matching pair, for $0 \leq i \leq t'$.
\end{enumerate}
\end{defi}

\begin{game}[Hereditary \ms Game with Repebbling]\label{msgame2}
Define the $r$-round, $k$-color \emph{hereditary \ms game with repebbling on $(\cA, \cB)$} as identical to the $r$-round, $k$-color \ms game with repebbling on $(\cA, \cB)$, except that, at the end of each of $r$ rounds, Duplicator needs to exhibit a pebbled structure on the left and a pebbled structure on the right forming a hereditary match (Definition \ref{def:hereditarymatch}), while Spoiler needs to exhibit a configuration within $r$ rounds where no pebbled structure on the left forms a hereditary match with any pebbled structure on the right.
\end{game}

Note that a Duplicator win in round $r$ of Game \ref{msgame2} immediately certifies a sequence of matching pairs that are valid for each of the previous rounds, and so if Duplicator wins Game \ref{msgame2} on an instance $(\cA, \cB)$, she certainly wins Game \ref{msgame1} on $(\cA, \cB)$. We also note that when $k = r$, both Games \ref{msgame1} and \ref{msgame2} are identical to ordinary \ms games. Furthermore, it is straightforward to see that Duplicator's oblivious strategy is optimal for both these variants.

We observe that the instance from Lemma \ref{weak_indistinguishabiliity_lemma} does not work as a counterexample for Game \ref{msgame2}, as Spoiler does win the $3$-round, $2$-color hereditary \ms game with repebbling on $(\{\lo(3)\},\{\lo(2)\})$. We leave this to the reader to verify.

However, it turns out that this stronger candidate game also fails to simultaneously capture the number of quantifiers and number of variables for {\fo}-distinguishability. This will follow from Proposition \ref{prop:hms-failure-to-capture}, but we will first need to set up stronger methods in order to prove such a result.

We now introduce a game that \textit{does} simultaneously capture these two parameters.

\subsection{The Quantifier-Variable Tree Game}

We define the \emph{$(r, k)$-quantifier-variable tree (QVT) game}, denoted the $\mathcal{QVT}(r,k)$ game, on $(\cA, \cB)$ as follows.

Two players, Spoiler and Duplicator, play by growing a \emph{game tree} $\cT$, starting from a single root node $\Xr$. Once again, we have a set $\cC$ of \emph{pebble colors}, with $|\cC| = k$, and arbitrarily many pebbles available of each color. Throughout the rest of this section, name the colors in $\cC$ as $x_1, \ldots, x_k$. The color $x_i$ will turn out to correspond to the variable $x_i$, and so we use this ambiguous notation rather deliberately.

We consider the leaf nodes in $\cT$ to be \emph{open} or \emph{closed}, with the root $\Xr$ being considered an open leaf at the start of the game. We say $\cT$ is \emph{closed} if all its leaves are closed.

Each node of $\cT$ consists of a \emph{left} side, a \emph{right} side, and a \emph{counter} $r' \in \N$. Each side consists of a set of pebbled $\tau$-structures. We denote a node in $\cT$ as a tuple $\gnode{\text{left}}{\text{right}}{r'}$.

At the root node $\Xr \in V(\cT)$,
\begin{itemize}
\item the left side consists of $\cA$, viewed as pebbled $\tau$-structures (with no pebbles placed yet).
\item the right side consists of $\cB$, viewed as pebbled $\tau$-structures (with no pebbles placed yet).
\item the counter $r'$ is set to $r$.
\end{itemize}
The root node, therefore, is denoted $\Xr = \gnode{\cA}{\cB}{r}$.

The contents $\gnode{\text{left}}{\text{right}}{r'}$ of the node $X \in V(\cT)$ will correspond to a \emph{configuration} of the \qvt game. When the context is clear, we will often identify a $\cT$-node by its configuration. Throughout the tree $\cT$, we will maintain the invariant that on every node $X \in V(\cT)$, a pebble color in $\cC$ appears in one pebbled structure iff it appears in every pebbled structure throughout the configuration.

Spoiler on his turn can perform any of the following moves:

\begin{itemize}
\item $\pebbleleft$ :
  \begin{enumerate}
  \item Spoiler chooses an open leaf node $X = \gnode{\cA'}{\cB'}{r'}$ with $r' \geq 1$, and a pebble color $x_i \in \cC$. If all structures in $X$ contain the pebble color $x_i$, Spoiler removes all of these pebbles.
  \item For each pebbled structure $\bA \in \cA'$, Spoiler places a pebble colored $x_i$ on an element in the universe of $\bA$. Call this new set of pebbled structures $\cA''$.
  \item For each pebbled structure $\bB \in \cB'$, Duplicator \emph{may} make any number of copies of $\bB$, then \emph{must} place a pebble colored $x_i$ on an element in the universe of $\bB$ and an element in the universe of each copy. Call this new set of pebbled structures $\cB''$.
   \item Spoiler makes a new open leaf
     $X' = \gnode{\cA''}{\cB''}{r'-1}$ in $\cT$, with
     parent $X$.
     Note that
     $X$ is no longer a leaf
     in $\cT$.
   \end{enumerate}
 \item $\pebbleright$: This move is dual to $\pebbleleft$; Spoiler plays on $\cB$, and Duplicator responds on $\cA$.
 \item $\splitleft$:
   \begin{enumerate}
   \item Spoiler chooses an open leaf node $X = \gnode{\cA'}{\cB'}{r'}$ and $r'_1,r'_2 \in \N$ such that $r' = r'_1 + r'_2$.
   \item Spoiler partitions $\cA'$ so that $\cA' = \cA_1' \cup \cA_2'$.
   \item Spoiler makes two new open leaf nodes $X_1 = \gnode{\cA_1'}{\cB'}{r'_1}$ and $X_2 = \gnode{\cA_2'}{\cB'}{r'_2}$ in $\cT$. Both new nodes have parent $X$.
    \end{enumerate}
  \item $\splitright$: This move is dual to
    $\splitleft$; Spoiler partitions $\cB'$.
  \item $\swap$: Spoiler chooses an open leaf node $X = \gnode{\cA'}{\cB'}{r'}$ in $\cT$. He makes a new open leaf node $X' = \gnode{\cB'}{\cA'}{r'}$ in $\cT$, with parent $X$.
\item $\close$: Spoiler chooses an open leaf node $X$ in $\cT$, say $\gnode{\cA'}{\cB'}{r'}$. Suppose there is an atomic formula $\varphi \in \mathrm{FO}^k$ with variables $\{x_{i_1}, \ldots, x_{i_m}\}$, such that:
\begin{enumerate}
    \item Every pebbled structure $\langle\bA ~|~ a_1, \ldots, a_{\ell}\rangle$ in $\cA$ has elements $a_{i_1}, \ldots, a_{i_m} \in A$ pebbled by colors $x_{i_1}, \ldots, x_{i_m}$ respectively, and $\bA \models \varphi(x_{i_1}/a_{i_1}, \ldots, x_{i_m}/a_{i_m})$.
    \item Every pebbled structure $\langle\bB ~|~ b_1, \ldots, b_{\ell}\rangle$ in $\cB$ has elements $b_{i_1}, \ldots, b_{i_m} \in B$ pebbled by colors $x_{i_1}, \ldots, x_{i_m}$ respectively, and $\bB \models \lnot\varphi(x_{i_1}/b_{i_1}, \ldots, x_{i_m}/b_{i_m})$.
\end{enumerate}
Then, Spoiler can mark $X$ \emph{closed}.
\end{itemize}
We say that $\cT$ is \emph{closed} if there are no open leaf nodes. Spoiler wins the game if he closes $\cT$, and Duplicator wins otherwise.

The following theorem characterizes the expressive power of the \qvt game, showing that it does indeed capture simultaneous bounds on number of quantifiers and number of variables. The key intuition behind the proof is that a closed game tree $\cT$ is by design isomorphic to the parse tree for a separating sentence for $(\cA,\cB)$, with the moves $\pebbleleft$, $\pebbleright$, $\splitleft$, $\splitright$, and $\swap$ in the game tree corresponding respectively to $\exists$, $\forall$, $\lor$, $\land$, and $\lnot$ in the separating sentence. We omit the proof, as it will follow immediately from the more general Theorem \ref{thm:syntacticchar} and Example \ref{exa:quantcount} in Section \ref{sec:syntactic}.

\begin{thm}[Equivalence theorem for QVT]\label{thm:qvtchar}
Spoiler has a winning strategy for the $\QVT{k}{r}{\cA}{\cB}$ iff there is an $r$-quantifier, $k$-variable sentence over $\tau$ that is separating for $(\cA,\cB)$.
\end{thm}

For ease of arguments, we will henceforth refer to $\pebbleleft$ or $\pebbleright$ moves as \emph{pebble moves}, and $\splitleft$ or $\splitright$ moves as \emph{split moves}. We will say an internal vertex of $\cT$ is \emph{closed} when all $\cT$-leaves descended from it are closed. Note that the value of the counter $r'$ at each $\cT$-node effectively represents the ``budget'' on the number of pebble moves available to Spoiler to close that node.

It is not entirely obvious that Duplicator's optimal strategy in the \qvt game is also to play the oblivious strategy. The following proposition asserts that this is still true. The proof, once again, is omitted, as it will follow from the more general Proposition \ref{prop:obliviousST}.

\begin{prop}\label{prop:obliviousQVT}
If Duplicator has a winning strategy in the $\qvt$ game, then the oblivious strategy is a winning strategy.
\end{prop}

Figure \ref{fig_qvt_example} depicts a complete game tree $\cT$ for the $(3, 2)$-\qvt game on $(\lo(3), \lo(2))$. Duplicator plays her oblivious strategy, but Spoiler is able to close the tree. Note that this is the same instance as both Figure \ref{fig:3_vs_2} and Lemma \ref{weak_indistinguishabiliity_lemma}.

\tikzset{RED/.style={draw=red, very thick, fill=none, font=\small, circle, inner sep=1mm}}
\tikzset{BLUE/.style={draw=blue, very thick, fill=none, font=\small, circle, inner sep=1mm}}
\tikzset{BLUESMALL/.style={draw=blue, very thick, fill=none, font=\small, circle, inner sep=0.6mm}}
\tikzset{BLANK/.style={draw=black, fill=none, font=\small, circle, inner sep=1mm}}
\tikzset{LABEL/.style={draw=none, fill=none, font=\small, circle, inner sep=0cm}}
\tikzset{COUNTER/.style={draw=none, fill=none, font=\tiny, rectangle, color=teal, inner sep=1mm}}
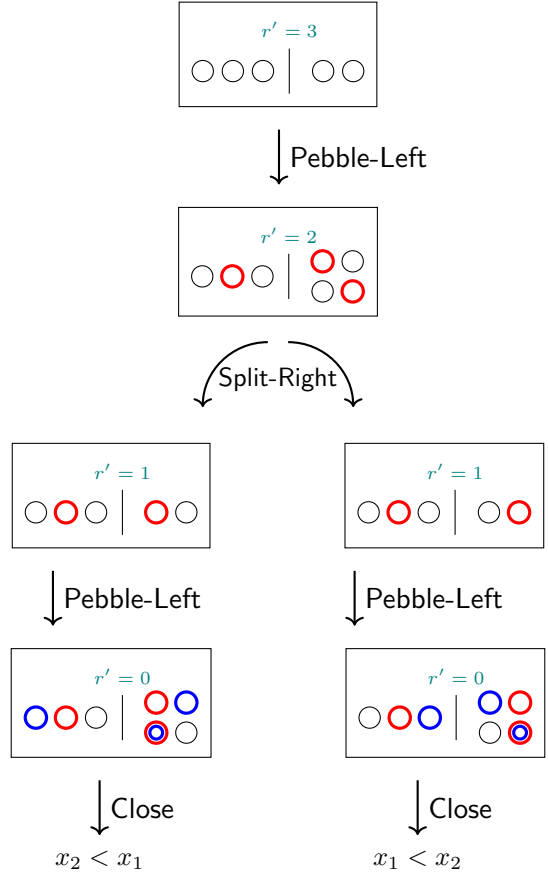
\begin{figure}[ht]
\centering
\begin{tikzpicture}[show background rectangle, scale=.5]

\node [BLANK] at (-2.3,0) {};
\node [BLANK] at (-1.5,0) {};
\node [BLANK] at (-0.7,0) {};

\draw (0, 0.6) -- (0, -0.6);
\node [COUNTER] at (0,1.1) {$r' = 3$};

\node [BLANK] at (0.9,0) {};
\node [BLANK] at (1.7,0) {};

\end{tikzpicture}

\begin{tikzpicture}[state/.style={rectangle,draw=none, rounded corners},arrow/.style={->,thick}]
  \node[state] (tail) at (0,0.5) {};
  \node[state] (head) at (0,-0.5) {};
  \node[state] (invis) at (-2,0) {};
  \draw[arrow] (tail) edge node[auto] {$\pebbleleft$} (head);
\end{tikzpicture}

\begin{tikzpicture}[show background rectangle, scale=.5]

\node [BLANK] at (-2.3,0) {};
\node [RED] at (-1.5,0) {};
\node [BLANK] at (-0.7,0) {};

\draw (0, 0.6) -- (0, -0.6);
\node [COUNTER] at (0,1.1) {$r' = 2$};

\node [RED] at (0.9,0.4) {};
\node [BLANK] at (1.7,0.4) {};

\node [BLANK] at (0.9,-0.4) {};
\node [RED] at (1.7,-0.4) {};
\end{tikzpicture}

\begin{tikzpicture}[state/.style={rectangle,draw=none, rounded corners},arrow/.style={->,thick}]
  \node[state] (tail) at (0,0.5) {};
  \node[state] (head1) at (-1,-0.5) {};
  \node[state] (head2) at (1,-0.5) {};
  \node [LABEL] at (0,0) {$\splitright$};
  \draw[arrow] (tail) edge[out=180,in=90] node[auto] {} (head1);
  \draw[arrow] (tail) edge[out=0,in=90] node[auto] {} (head2);
\end{tikzpicture}

\begin{tikzpicture}[show background rectangle, scale=.5]

\node [BLANK] at (-2.3,0) {};
\node [RED] at (-1.5,0) {};
\node [BLANK] at (-0.7,0) {};

\draw (0, 0.6) -- (0, -0.6);
\node [COUNTER] at (0,1.1) {$r' = 1$};

\node [RED] at (0.9,0) {};
\node [BLANK] at (1.7,0) {};
\end{tikzpicture}
\hspace{0.6in}
\begin{tikzpicture}[show background rectangle, scale=.5]

\node [BLANK] at (-2.3,0) {};
\node [RED] at (-1.5,0) {};
\node [BLANK] at (-0.7,0) {};

\draw (0, 0.6) -- (0, -0.6);
\node [COUNTER] at (0,1.1) {$r' = 1$};

\node [BLANK] at (0.9,0) {};
\node [RED] at (1.7,0) {};
\end{tikzpicture}

\begin{tikzpicture}[state/.style={rectangle,draw=none, rounded corners},arrow/.style={->,thick}]
  \node[state] (tail) at (0,0.5) {};
  \node[state] (head) at (0,-0.5) {};
  \draw[arrow] (tail) edge node[auto] {$\pebbleleft$} (head);
\end{tikzpicture}
\hspace{0.6in}
\begin{tikzpicture}[state/.style={rectangle,draw=none, rounded corners},arrow/.style={->,thick}]
  \node[state] (tail) at (0,0.5) {};
  \node[state] (head) at (0,-0.5) {};
  \draw[arrow] (tail) edge node[auto] {$\pebbleleft$} (head);
\end{tikzpicture}

\begin{tikzpicture}[show background rectangle, scale=.5]

\node [BLUE] at (-2.3,0) {};
\node [RED] at (-1.5,0) {};
\node [BLANK] at (-0.7,0) {};

\draw (0, 0.6) -- (0, -0.6);
\node [COUNTER] at (0,1.1) {$r' = 0$};

\node [RED] at (0.9,0.4) {};
\node [BLUE] at (1.7,0.4) {};

\node [RED] at (0.9,-0.4) {};
\node [BLANK] at (1.7,-0.4) {};
\node [BLUESMALL] at (0.9, -0.4) {};
\end{tikzpicture}
\hspace{0.6in}
\begin{tikzpicture}[show background rectangle, scale=.5]

\node [BLANK] at (-2.3,0) {};
\node [RED] at (-1.5,0) {};
\node [BLUE] at (-0.7,0) {};

\draw (0, 0.6) -- (0, -0.6);
\node [COUNTER] at (0,1.1) {$r' = 0$};

\node [BLUE] at (0.9,0.4) {};
\node [RED] at (1.7,0.4) {};

\node [BLANK] at (0.9,-0.4) {};
\node [RED] at (1.7,-0.4) {};
\node [BLUESMALL] at (1.7, -0.4) {};
\end{tikzpicture}

\begin{tikzpicture}[state/.style={rectangle,draw=none, rounded corners},arrow/.style={->,thick}]
  \node[state] (tail) at (0,0.5) {};
  \node[state] (head) at (0,-0.5) {};
  \draw[arrow] (tail) edge node[auto] {$\close$} (head);
  \node [LABEL] at (0,-0.7) {$x_2 < x_1$};
\end{tikzpicture}
\hspace{0.6in}
\begin{tikzpicture}[state/.style={rectangle,draw=none, rounded corners},arrow/.style={->,thick}]
  \node[state] (tail) at (0,0.5) {};
  \node[state] (head) at (0,-0.5) {};
  \node[state] (invis) at (-1.2,0) {};
  \draw[arrow] (tail) edge node[auto] {$\close$} (head);
  \node [LABEL] at (0,-0.7) {$x_1 < x_2$};
\end{tikzpicture}
\caption{A complete play of the $\QVT{2}{3}{\{\lo(3)\}}{\{\lo(2)\}}$. Spoiler closes the game tree $\cT$ using only three pebble moves, and two pebble colors $x_1$ (red) and $x_2$ (blue). The atomic separating sentences are shown below, and the resulting separating sentence is $\exists x_1(\exists x_2 (x_2 < x_1) \land \exists x_2(x_1 < x_2))$, which can be read off $\cT$.}
\label{fig_qvt_example}
\end{figure}

\subsection{The hereditary \ms game with repebbling does not capture number of variables}

We are now ready to prove that the $r$-round, $k$-color hereditary \ms game with repebbling (Game \ref{msgame2}) does \emph{not} capture the number of variables and quantifiers simultaneously, using Theorem \ref{thm:qvtchar}.

\begin{restatable}{prop}{hmsdoesnotwork}\label{prop:hms-failure-to-capture}
The following are both true:
\begin{enumerate}
\item There is no $3$-quantifier, $2$-variable separating sentence for $(\lo(4), \lo(3))$.
\item Spoiler has a winning strategy in the $3$-round, $2$-color hereditary \ms game with repebbling on $\{\{\lo(4)\}, \{\lo(3)\}\}$.
\end{enumerate}
\end{restatable}

\begin{proof}
We first show (a), i.e.,~there is no $3$-quantifier, $2$-variable separating sentence for $(\lo(4), \lo(3))$. By Theorem \ref{thm:qvtchar}, it suffices to show that Duplicator has a winning strategy on the game $\QVT{2}{3}{\{\lo(4)\}}{\{\lo(3)\}}$. Since we are starting with singleton sets, Spoiler cannot play a split move, and gains nothing from a $\swap$ move at the root node of $\cT$. Let us consider each of the possible pebble moves one by one. Note that the counter will decrement after this move from $3$ to $2$.

It is straightforward to check that, at the root node of $\cT$, a $\pebbleright$ play on L1 or L3 is met easily with a Duplicator response on B1 or B4 respectively. This results in a Duplicator victory in the $3$-round \ms game, and therefore in $\QVT{3}{3}{\{\lo(4)\}}{\{\lo(3)\}}$, and therefore also in $\QVT{2}{3}{\{\lo(4)\}}{\{\lo(3)\}}$. So assume Spoiler plays $\pebbleright$ on L2 at the root of $\cT$. By Proposition \ref{prop:obliviousQVT}, we can assume Duplicator responds obliviously, and we are again at the situation depicted in Figure \ref{fig:4_v_3_MSr1}. At this point, a Spoiler split move can be ruled out, as the $\cT$-node resulting from the split containing the second structure on the left will require two more pebble moves to close, while the other $\cT$-node from the split will require at least one more pebble move to close. Spoiler therefore has to use up his second pebble move at this point, and we examine in turn the two possible pebble moves at this $\cT$-node. After this move, the counter will decrement to $1$.

If for his second pebble move, Spoiler plays $\pebbleright$ on L1, then Duplicator can respond on B1 on the second board on the left, and survive another pebble move on just this pair of structures. A symmetric argument applies when Spoiler plays $\pebbleright$ on L3. Finally, it is easy to see that Spoiler playing $\pebbleright$ on L2 achieves nothing. So Spoiler must play $\pebbleleft$ for his second pebble move.

Let us now examine this second pebble move by Spoiler on the second structure on the left. Where does Spoiler place a pebble? If he plays on B1 or B2, Duplicator responds with L1 or L2 respectively, and easily survives another pebble move. So, Spoiler must play either B3 or B4, which must be met with a Duplicator response on L3. In fact, for Spoiler to win on this pair of pebbled structures with one more pebble move, he must play on B3 (not B4) for his second pebble move. A symmetric argument shows that on the third structure on the left, Spoiler must play on B2 in his second pebble move. We can now look at Figure \ref{fig:4_v_3_MSr2-2}, where we have only depicted four relevant pebbled structures after two pebble moves.

\tikzset{RED/.style={draw=red, very thick, fill=none, font=\tiny, circle, inner sep=0.5mm}}
\tikzset{BLUE/.style={draw=blue, very thick, fill=none, font=\tiny, circle, inner sep=0.5mm}}
\tikzset{BLUESMALL/.style={draw=blue, very thick, fill=none, font=\tiny, circle, inner sep=0.3mm}}
\tikzset{BLANK/.style={draw=black, fill=none, font=\tiny, circle, inner sep=0.5mm}}
\tikzset{LABEL/.style={draw=none, fill=none, font=\small, circle, inner sep=0cm}}

\begin{figure}[ht]
\centering
\begin{tikzpicture}

\node [BLANK] at (-3.1,0.5) {B1};
\node [RED] at (-2.3,0.5) {B2};
\node [BLUE] at (-1.5,0.5) {B3};
\node [BLANK] at (-0.7,0.5) {B4};

\node [BLANK] at (-3.1,-0.5) {B1};
\node [BLUE] at (-2.3,-0.5) {B2};
\node [RED] at (-1.5,-0.5) {B3};
\node [BLANK] at (-0.7,-0.5) {B4};

\draw (0.2, 0.6) -- (0.2, -0.6);

\node [BLANK] at (1,0.5) {L1};
\node [RED] at (1.8,0.5) {L2};
\node [BLUE] at (2.6,0.5) {L3};

\node [BLUE] at (1,-0.5) {L1};
\node [RED] at (1.8,-0.5) {L2};
\node [BLANK] at (2.6,-0.5) {L3};

\end{tikzpicture}
\caption{Configuration (partial) after two pebble moves, depicting just two pairs of pebbled structures that maintain isomorphisms.}
\label{fig:4_v_3_MSr2-2}
\end{figure}
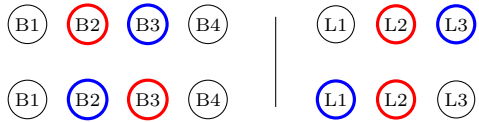

At this point, splitting is not an option, as the counter is at $1$, and so splitting would cause one of the two resultant $\cT$-nodes to have the counter value at $0$, and Spoiler cannot close that node. But Spoiler cannot win without splitting either, as moving any pebble on any of the four structures in Figure \ref{fig:4_v_3_MSr2-2} can be met with a response on at least one of the two structures on the other side that maintains an isomorphism. It follows that Spoiler cannot win with a $\pebbleright$ move at the root of $\cT$.

It is straightforward to check that, at the root of $\cT$, a $\pebbleleft$ play on B1 or B4 is met easily with a Duplicator response on L1 or L3 respectively. So assume Spoiler plays $\pebbleleft$ on B2 for his first pebble move (B3 will be symmetric). By Proposition \ref{prop:obliviousQVT}, we can assume Duplicator responds obliviously, as depicted in Figure \ref{fig:vs3-rd1}. The counter is at $2$.

\tikzset{RED/.style={draw=red, very thick, fill=none, font=\tiny, circle, inner sep=0.5mm}}
\tikzset{BLUE/.style={draw=blue, very thick, fill=none, font=\tiny, circle, inner sep=0.5mm}}
\tikzset{BLUESMALL/.style={draw=blue, very thick, fill=none, font=\tiny, circle, inner sep=0.3mm}}
\tikzset{BLANK/.style={draw=black, fill=none, font=\tiny, circle, inner sep=0.5mm}}
\tikzset{LABEL/.style={draw=none, fill=none, font=\small, circle, inner sep=0cm}}

\begin{figure}[ht]
\centering
\begin{tikzpicture}

\node [BLANK] at (-3.1,0) {B1};
\node [RED] at (-2.3,0) {B2};
\node [BLANK] at (-1.5,0) {B3};
\node [BLANK] at (-0.7,0) {B4};

\draw (0.2, 0.6) -- (0.2, -0.6);

\node [RED] at (1,1) {L1};
\node [BLANK] at (1.8,1) {L2};
\node [BLANK] at (2.6,1) {L3};

\node [BLANK] at (1,0) {L1};
\node [RED] at (1.8,0) {L2};
\node [BLANK] at (2.6,0) {L3};

\node [BLANK] at (1,-1) {L1};
\node [BLANK] at (1.8,-1) {L2};
\node [RED] at (2.6,-1) {L3};

\end{tikzpicture}
\caption{Configuration of the game $\QVT{2}{3}{\{\lo(4)\}}{\{\lo(3)\}}$ after the first pebble move, where Spoiler plays $\pebbleleft$ on B2 and Duplicator responds obliviously.}
\label{fig:vs3-rd1}
\end{figure}
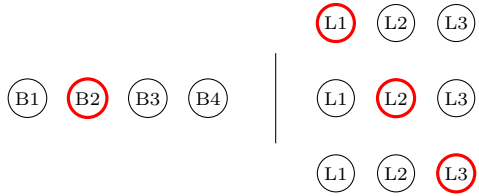

By a similar argument as before, we can conclude that a Spoiler split move can be ruled out. We can now argue that in order to win in three pebble moves, Spoiler must play $\pebbleright$ for his second pebble move. Indeed, if instead, Spoiler plays $\pebbleleft$, then a play of
\begin{itemize}
\item B1 is met with Duplicator playing on L1 on the middle structure on the right and surviving one more pebble move on that pair;
\item B3 is met with Duplicator playing L2 on the top structure on the right and L3 on the middle structure on the right, and surviving one more pebble move on one of these structures;
\item B4 is met with Duplicator playing L3 on both of the top two structures on the right, and surviving one more pebble move on one of these structures.
\end{itemize}

We now claim that for Spoiler's second pebble move with the $\pebbleright$, he must play on L3 on the middle structure on the right. Otherwise, a play on L1 or L2 is met by a Duplicator move on B1 or B2 respectively on the structure on the left, where she easily survives one more pebble move just on this pair. However, a Spoiler play of L3 is met with a Duplicator play on B4 on the left, and she can survive one more pebble move on this pair. This concludes the proof of (a).

We next show (b), i.e.,~Spoiler wins the $3$-round, $2$-color hereditary \ms game with repebbling on $\{\{\lo(4)\}, \{\lo(3)\}\}$. Spoiler's first round $\playright$ move, with Duplicator's oblivious responses on the left side, is shown in Figure \ref{fig:4_v_3_MSr1}.

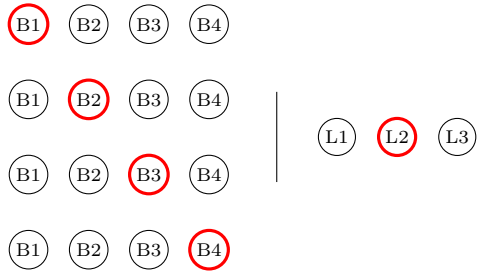
\begin{figure}[ht]
\centering
\begin{tikzpicture}

\node [RED] at (-3.1,1.5) {B1};
\node [BLANK] at (-2.3,1.5) {B2};
\node [BLANK] at (-1.5,1.5) {B3};
\node [BLANK] at (-0.7,1.5) {B4};

\node [BLANK] at (-3.1,0.5) {B1};
\node [RED] at (-2.3,0.5) {B2};
\node [BLANK] at (-1.5,0.5) {B3};
\node [BLANK] at (-0.7,0.5) {B4};

\node [BLANK] at (-3.1,-0.5) {B1};
\node [BLANK] at (-2.3,-0.5) {B2};
\node [RED] at (-1.5,-0.5) {B3};
\node [BLANK] at (-0.7,-0.5) {B4};

\node [BLANK] at (-3.1,-1.5) {B1};
\node [BLANK] at (-2.3,-1.5) {B2};
\node [BLANK] at (-1.5,-1.5) {B3};
\node [RED] at (-0.7,-1.5) {B4};

\draw (0.2, 0.6) -- (0.2, -0.6);

\node [BLANK] at (1,0) {L1};
\node [RED] at (1.8,0) {L2};
\node [BLANK] at (2.6,0) {L3};

\end{tikzpicture}
\caption{Configuration after round $1$ of the $3$-round, $2$-color hereditary \ms game with repebbling on $\{\{\lo(4)\}, \{\lo(3)\}\}$.}
\label{fig:4_v_3_MSr1}
\end{figure}

Spoiler will play $\playleft$ for the two remaining rounds. His second round plays are indicated in blue in Figure \ref{fig:4_v_3_MSr2}.

\begin{figure}[ht]
\centering
\begin{tikzpicture}

\node [RED] at (-3.1,1.5) {B1};
\node [BLUE] at (-2.3,1.5) {B2};
\node [BLANK] at (-1.5,1.5) {B3};
\node [BLANK] at (-0.7,1.5) {B4};

\node [BLANK] at (-3.1,0.5) {B1};
\node [RED] at (-2.3,0.5) {B2};
\node [BLUE] at (-1.5,0.5) {B3};
\node [BLANK] at (-0.7,0.5) {B4};

\node [BLANK] at (-3.1,-0.5) {B1};
\node [BLUE] at (-2.3,-0.5) {B2};
\node [RED] at (-1.5,-0.5) {B3};
\node [BLANK] at (-0.7,-0.5) {B4};

\node [BLANK] at (-3.1,-1.5) {B1};
\node [BLANK] at (-2.3,-1.5) {B2};
\node [BLUE] at (-1.5,-1.5) {B3};
\node [RED] at (-0.7,-1.5) {B4};

\draw (0.2, 0.6) -- (0.2, -0.6);

\node [BLANK] at (1,0) {L1};
\node [RED] at (1.8,0) {L2};
\node [BLANK] at (2.6,0) {L3};

\end{tikzpicture}
\caption{Configuration (partial) after Spoiler's second round $\playleft$ move (given in blue) for the $3$-round, $2$-color hereditary \ms game with repebbling on $\{\{\lo(4)\}, \{\lo(3)\}\}$.}
\label{fig:4_v_3_MSr2}
\end{figure}
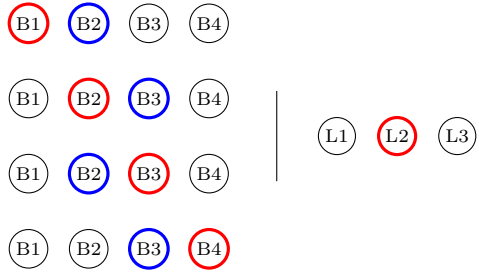

In response, let us WLOG only consider Duplicator's plays on L1 and L3 (the L2 move achieves nothing). Responding with L1 on the right keeps an isomorphism going with pebbled structures 3 and 4 on the left, and responding with L3 on the right keeps an isomorphism going with pebbled structures 1 and 2 on the left. In response, in the third round, Spoiler plays the red pebble on B3, B4, B1, and B2 on pebbled structures 1, 2, 3, and 4 respectively, in all cases breaking all hereditary isomorphisms with pebbled structures on the right. This concludes the proof.
\end{proof}

It is worth noting that we could \textit{not} have proven Proposition \ref{prop:hms-failure-to-capture} via a straightforward pebble game argument with $3$ rounds and $2$ pebbles. Indeed, Spoiler easily wins such a game beginning with a play on B2, B3 or L2. This implies that there is a rank-$3$, $2$-variable separating sentence for $(\{\lo(4)\},\{\lo(3)\})$, namely, $\exists x ( (\exists y (y<x)) \land \exists y ( x < y  \land \exists x (y <x)))$.
\section{The Syntactic Game}\label{sec:syntactic}

Consider a node $X = \gnode{\cA'}{\cB'}{r'}$ in the game tree $\cT$ for a given \qvt game. The counter $r' \in \N$ in $X$ is the number of pebble moves Spoiler may use to close the subtree of $\cT$ rooted at $X$. The proof of Theorem \ref{thm:qvtchar} will follow from the fact that this is exactly the number of quantifiers needed in some sense to ``distinguish'' between $\cA'$ and $\cB'$. In fact, as stated earlier, we will prove a significantly stronger generalization. We first need the formal notion of a \emph{syntactic measure}. Such measures were introduced to derive combinatorial games for ``reasonable'' complexity measure of infinitary formulas in the logic $\cL_{\omega_1, \omega}$ \cite[Definition 5.1]{DBLP:journals/mlq/VaananenW13}. In this section, we will derive games precisely capturing many familiar complexity measures, including: quantifier count, quantifier rank, and formula size. The main idea in all of these games would be to have specific types of \emph{pebble moves} corresponding to the inductive ways to build up FO formulas; Spoiler will use these moves to explicitly simulate a separating formula, and will try to do this and close the game tree without running out of a ``budget'' on the syntactic measure, which we keep track of by means of a counter. In these moves, the number of pebble colors captures the number of variables.  For each logical connective occurring in the language in question, there is a move and a corresponding cost of using that operator. For example, to count the quantifier number, we add $1$ for each quantifier used, as in the \qvt game.  Thus, these games capture the descriptive cost function under consideration. They are the natural tool for proving descriptive lower bounds; sometimes they are the best choice for proving upper bounds as well.

\subsection{Compositional syntactic measures}

Recall that $\mathrm{FO}^k(\tau)$ is the set of well-formed {\fo} formulas over the schema $\tau$, that use only the variables $\{x_1, \ldots, x_k\}$. Note that $\mathrm{FO}^k(\tau)$ is defined inductively, where each formula is the
application of a logical symbol to some subformulas. A \emph{compositional syntactic measure} uses such an inductive definition to measure formula complexity.

Recall that $\mathsf{Pred}(\tau)$ is the set of predicate symbols of $\tau$, assumed to be finite.

\begin{defi}\label{def:syntacticmeasure}
A \emph{compositional syntactic measure} is a function $f : \mathrm{FO}^k(\tau) \to \N$ defined in terms of helper
functions $h_\lnot, h_\exists, h_\forall : \N \to \N$, $h_\lor, h_\land : \N^2 \to \N$, $h_{\mathsf{atomic}} : \mathsf{Pred}(\tau) \to \N$, as follows:  
\begin{itemize}
\item if $\varphi$ is atomic, say $\varphi = P(t_1,\ldots, t_r)$ for some $P \in \mathsf{Pred}(\tau)$, and terms $t_1, \ldots, t_r$, where $r$ is the arity of $P$, then $f(\varphi) = h_{\mathsf{atomic}}(P)$. 
\item if $\varphi = \neg \psi$, then $f(\varphi) = h_\lnot(f(\psi))$.
\item if $\varphi = (\psi \lor \gamma)$, then $f(\varphi) = h_\lor(f(\psi), f(\gamma))$.
\item if $\varphi = (\psi \land \gamma)$, then $f(\varphi) = h_\land(f(\psi), f(\gamma))$.
\item if $\varphi = \exists x ( \psi)$, then $f(\varphi) = h_\exists(f(\psi))$.
\item if $\varphi = \forall x (\psi)$, then $f(\varphi) = h_\forall(f(\psi))$.
\end{itemize}
\end{defi}

We call such a measure ``compositional'' because the value of $f$ is determined by the value of $f$ on its subformulas. Compositional syntactic measures are a large and natural class of syntactic measures, which includes quantifier rank, quantifier number, formula size, and number of atomic formulas. However, we remark that the number of distinct variables is \emph{not} a compositional measure, as the value of this measure on, say, a disjunction depends not only on the values on the subformulas, but actually on the subformulas themselves. We also note that in Definition \ref{def:syntacticmeasure}, we treat equality as an ordinary binary predicate.

We could have stated Definition \ref{def:syntacticmeasure} purely in terms of well-formed formulas over $\tau$, instead of the slice $\mathrm{FO}^k(\tau)$. However, we chose our convention to evoke the idea of fixing $k$ variables at the outset, which will be relevant in the rest of this section.

The following examples show that quantifier number, quantifier rank, and formula size are indeed compositional syntactic measures.

\begin{exa}\label{exa:quantcount}
The \emph{quantifier count} $f_q$ satisfies:
\begin{itemize}
\item if $\varphi$ is atomic,
  $f_q(\varphi) = 0$. In other words, $h_{\mathsf{atomic}}(\varphi) = 0$.
\item if $\varphi = \neg \psi$,
  $f_q(\varphi) = f_q(\psi)$. In other words, $h_\lnot(n) = n$.
\item if $\varphi = (\psi \land \gamma)$ or $\varphi = (\psi \lor \gamma)$, $f_q(\varphi) = f_q(\psi) + f_q(\gamma)$. In other words, $h_\lor(m, n) = h_\land(m, n) = m + n$.
\item if $\varphi = \exists x \psi$ or $\varphi = \forall x \psi$, $f_q(\varphi) = 1 + f_q(\psi)$. In other words, $h_\exists(n) = h_\forall(n) = n + 1$.
\end{itemize}
\end{exa}

\begin{exa}\label{exa:quantrank}
The \emph{quantifier rank} $f_r$ satisfies:
\begin{itemize}
\item if $\varphi$ is atomic,
  $f_r(\varphi) = 0$. In other words, $h_{\mathsf{atomic}}(\varphi) = 0$.
\item if $\varphi = \neg \psi$,
  $f_r(\varphi) = f_r(\psi)$. In other words, $h_\lnot(n) = n$.
\item if $\varphi = (\psi \land \gamma)$ or $\varphi = (\psi \lor \gamma)$, $f_r(\varphi) = \max(f_r(\psi), f_r(\gamma))$. In other words, $h_\lor(m, n) = h_\land(m, n) = \max(m, n)$.
\item if $\varphi = \exists x \psi$ or $\varphi = \forall x \psi$, $f_r(\varphi) = 1 + f_r(\psi)$. In other words, $h_\exists(n) = h_\forall(n) = n + 1$.
\end{itemize}
\end{exa}

\begin{exa}\label{exa:formulasize}
The \emph{formula size} $f_s$ satisfies:
\begin{itemize}
\item if $\varphi$ is atomic,
  $f_r(\varphi) = 1$. In other words, $h_{\mathsf{atomic}}(\varphi) = 1$.
\item if $\varphi = \neg \psi$,
  $f_r(\varphi) = 1 + f_r(\psi)$. In other words, $h_\lnot(n) = n + 1$.
\item if $\varphi = (\psi \land \gamma)$ or $\varphi = (\psi \lor \gamma)$, $f_r(\varphi) = 1 + f_r(\psi) + f_r(\gamma)$. In other words, $h_\lor(m, n) = h_\land(m, n) = 1 + m + n$.
\item if $\varphi = \exists x \psi$ or $\varphi = \forall x \psi$, $f_r(\varphi) = 1 + f_r(\psi)$. In other words, $h_\exists(n) = h_\forall(n) = n + 1$.
\end{itemize}
\end{exa}

We are now ready to define the syntactic game.

\subsection{The syntactic game}

Fix a $\tau$-structure $\bA$, and let $\alpha_\bA$ be a \emph{partial assignment}, i.e.,~a function that maps a set of variables to elements in $A$. We call the ordered pair $(\bA, \alpha_\bA)$ a \emph{structure-assignment pair}. Note that if we identify pebble colors with variables, a pebbled $\tau$-structure $\langle\bA ~|~ a_1, \ldots, a_t\rangle$ has a one-to-one correspondence with the pair $(\bA, \alpha_\bA)$, where $\alpha_\bA$ maps the variable corresponding to a pebble color to the element in $A$ that it is placed on. In particular, a $\tau$-structure $\bA$ has a one-to-one correspondence with the pair $(\bA, \varnothing)$, where $\varnothing$ is the empty assignment.

Just as a $\tau$-structure $\bA$ satisfies a sentence over $\tau$, we can extend this notion to a structure-assignment pair $(\bA, \alpha_\bA)$ satisfying a \emph{formula}, which is allowed to have free variables as long as all such variables are in the domain\footnote{ Note that throughout this section, the term \emph{domain} will always refer to the domain of the assignment functions, and never to the universe of $\tau$-structures.} of $\alpha_\bA$. Formally, let $(\bA, \alpha_\bA)$ be a structure-assignment pair, and $\psi$ be a formula such that $\mathsf{Free}(\psi) = \{x_{i_1}, \ldots, x_{i_m}\} \subseteq \mathsf{dom}(\alpha_\bA)$. Suppose $\alpha_\bA(x_{i_j}) = a_{i_j} \in A$ for all $1 \leq j \leq m$. We say that $(\bA, \alpha_\bA)$ \emph{satisfies} $\psi$ (or $(\bA, \alpha_\bA) \models \psi$) if $\bA \models \psi(x_{i_1}/a_{i_1}, \ldots, x_{i_m}/a_{i_m})$.

Of course, for every structure-assignment pair $(\bA, \alpha_\bA)$ and formula $\psi$ with $\mathsf{Free}(\psi) \subseteq \mathsf{dom}(\alpha_\bA)$, exactly one of $(\bA, \alpha_\bA) \models \psi$ and $(\bA, \alpha_\bA) \models \lnot\psi$ is true.

For the rest of this paper, by convention, we will denote sets of structure-assignment pairs with script typeface $\mathscr{A}, \mathscr{B}$, and so on. Note that the $\tau$-structures among the pairs in any such $\mathscr{A}$ can in general be different. We now state a definition for convenience.

\begin{defi}\label{def:domainconsistent}
Two sets $\sA$ and $\sB$ of structure-assignment pairs are called \emph{domain-consistent} if there is a set $W$ of variables such that for every structure-assignment pair $(\mathbf{X}, \alpha_{\mathbf{X}}) \in \sA\cup\sB$, we have that $\mathbf{dom}(\alpha_{\mathbf{X}}) = W$. We denote this common domain $W$ by $\mathsf{dom}(\sA)$ (or, equivalently, $\mathsf{dom}(\sB)$).
\end{defi}

In other words, every partial assignment across all pairs in the two sets have the same domain. In particular, every structure-assignment pair within the set $\sA$ (or, equivalently, $\sB$) also has the same domain for its partial assignment.

We can now formalize what it means for a formula with free variables to be a separating formula, thereby generalizing Definition \ref{def:sepsentence}.

\begin{defi}\label{def:sepformula}
Let $\sA$ and $\sB$ be domain-consistent sets of structure-assignment pairs with common domain $\mathsf{dom}(\sA)$. A \emph{separating formula for $(\sA,\sB)$} is a {\fo}-formula $\psi$ with $\mathsf{Free}(\psi) \subseteq \mathsf{dom}(\sA)$, such that every $(\bA, \alpha_\bA) \in \sA$ satisfies $(\bA, \alpha_\bA) \models \varphi$, and every $(\bB, \alpha_\bB) \in \sB$ satisfies $(\bB, \alpha_\bB) \models \lnot\varphi$.
\end{defi}

Note that this does generalize Definition \ref{def:sepsentence}, using the correspondence between $\bA$ and $(\bA, \varnothing)$. A separating sentence is just a separating formula with no free variables.

We are now ready to define the syntactic game. Fix two domain-consistent sets $\sA$ and $\sB$ of structure-assignment pairs. For a compositional syntactic measure $f$, define the \emph{$(r, k)$-syntactic game on $f$}, denoted $\SG{r}{k}{f}$, on $(\sA, \sB)$ as follows.

Two players, Spoiler and
Duplicator, play by growing a \emph{game tree} $\cT$, starting from a single root node $\Xr$. As in the \qvt game, we have a set $\cC = \{x_1, \ldots, x_k\}$ of pebble colors, with arbitrarily many pebbles available of each color. Once again, the color $x_i$ will correspond to the variable $x_i$. The leaf nodes in $\cT$ are open or closed, with the root $\Xr$ considered to be an open leaf at the start of the game. As before, we define $\cT$ to be \emph{closed} if all its leaves are closed. The notion of an internal vertex being closed also carries over from Section \ref{sec:finitepebbles}. The nodes $X \in V(\cT)$ will correspond to a tuple $\gnode{\text{left}}{\text{right}}{c}$ as before. We define a configuration of the game as before, and maintain the same invariant as in the \qvt game. Recall that a pebbled $\tau$-structure $\langle\bA ~|~ a_1, \ldots, a_t\rangle$ now has a one-to-one correspondence with the pair $(\bA, \alpha_\bA)$, where $\alpha_\bA$ is the assignment that maps a variable corresponding to a pebble color to the element in $A$ it is placed on. To avoid notational clutter, we will view $\sA$ or $\sB$ as either a set of pebbled $\tau$-structures or as a set of structure-assignment pairs interchangeably, where the precise one will be clear from context.

At the root node $\Xr \in V(\cT)$,
\begin{itemize}
\item the left side consists of $\sA$, viewed as pebbled $\tau$-structures (with the pebbling inherited from the assignment functions).
\item the right side consists of $\sB$, viewed as pebbled $\tau$-structures (with the pebbling inherited from the assignment functions).
\item the counter $r'$ is set to $r$.
\end{itemize}
The root node, therefore, is denoted $\Xr = \gnode{\sA}{\sB}{r}$. Note that since $\sA$ and $\sB$ are domain-consistent, every pebbled $\tau$-structure at the root has the same set of pebbles on it. We will maintain this invariant throughout every node of $\cT$ (i.e., every configuration will correspond to a pair of domain-consistent sets $\sA'$ and $\sB'$). As before, the root is considered an open leaf node at the start.

Once again, Spoiler on his turn can perform any of the following moves on an open leaf node $X = \gnode{\sA'}{\sB'}{r'}$, as long as the precondition for that move is met:

\begin{itemize}
\item $\pebbleleft$ : precondition: $h^{-1}_\exists(r') \neq \varnothing$
  \begin{enumerate}
  \item Spoiler chooses a pebble color $x_i \in \cC$. If all structures in $X$ contain the pebble color $x_i$, Spoiler removes all of these pebbles.
  \item For each pebbled structure $(\bA, \alpha_\bA) \in \sA'$, Spoiler places a pebble colored $x_i$ on an element in the universe of $\bA$. Call this new set of pebbled structures $\sA''$.
  \item For each pebbled structures $(\bB, \alpha_\bB) \in \sB'$, Duplicator \emph{may} make any number of copies of $(\bB, \alpha_\bB)$, then \emph{must} place a pebble colored $x_i$ on an element in the universe of $\bB$ and an element in the universe of each copy. Call this new set of pebbled structures $\sB''$. 
  \item Spoiler chooses some $r'' \in \N$ with $h_\exists(r'') = r'$. Note that this is guaranteed to exist by the precondition.
  \item Spoiler makes a new open leaf $X' = \gnode{\sA''}{\sB''}{r''}$ in $\cT$, with parent $X$. Note that $X$ is no longer a leaf in $\cT$.
   \end{enumerate}
 \item $\pebbleright$ : precondition: $h^{-1}_\forall(r') \neq \varnothing$

This move is dual to $\pebbleleft$; Spoiler plays on $\sB'$, Duplicator responds on $\sA'$, and Spoiler chooses some $r'' \in \N$ with $h_\forall(r'') = r'$.
 \item $\splitleft$: precondition: $h^{-1}_\lor(r') \neq \varnothing$
   \begin{enumerate}
   \item Spoiler partitions $\sA'$ so that $\sA' = \sA'_1 \cup \sA'_2$.
   \item Spoiler chooses $r'_1,r'_2 \in \N$ such that $h_\lor(r'_1, r'_2) = r'$. Note that these are guaranteed to exist by the precondition.
   \item Spoiler makes two new open leaf nodes $X_1 = \gnode{\sA'_1}{\sB'}{r'_1}$ and $X_2 = \gnode{\sA'_2}{\sB'}{r'_2}$ in $\cT$. Both new nodes have parent $X$.
    \end{enumerate}
 \item $\splitright$: precondition: $h^{-1}_\land(r') \neq \varnothing$

   This move is dual to $\splitleft$; Spoiler partitions $\sB'$ and chooses $r'_1,r'_2 \in \N$
   such that $h_\land(r'_1, r'_2) = r'$.
  \item $\swap$: precondition: $h^{-1}_\lnot(r') \neq \varnothing$.
  \begin{enumerate}
  \item Spoiler chooses $r'' \in \N$ such that $h_\lnot(r'') = r'$.
  \item Spoiler makes a new open leaf node $X' = \gnode{\sB'}{\sA'}{r''}$ in $\cT$, with parent $X$.
  \end{enumerate}
\item $\close$: If there is an atomic formula $\varphi$ over $\tau$ that is separating (Definition \ref{def:sepformula}) for $(\sA',\sB')$, such that $h_{\mathsf{atomic}}(\varphi) = r'$, Spoiler can mark $X$ \emph{closed}.
\end{itemize}
We say that $\cT$ is \emph{closed} if there are no open leaf nodes. Spoiler wins the game if he closes $\cT$, and Duplicator wins otherwise.

If $\cA$ and $\cB$ are sets of $\tau$-structures, we will abuse notation slightly and talk about the syntactic game on $(\cA, \cB)$, where these will be viewed as sets of structure-assignment pairs with the empty assignment throughout.

We are now ready to characterize the expressive power of the syntactic game. The key to the proof is the observation that a closed game tree $\cT$ is isomorphic to the parse tree of a formula separating
$(\sA,\sB)$.

\begin{thm}[Syntactic Game Fundamental Theorem]\label{thm:syntacticchar}
Fix a compositional syntactic measure $f$, and domain-consistent sets $\sA, \sB$ of structure-assignment pairs, with $\mathsf{dom}(\sA) \subseteq \{x_1, \ldots, x_k\}$. Spoiler has a winning strategy for the game $\SG{r}{k}{f}$ on $(\sA,\sB)$ if and only if there is a formula $\varphi \in \mathrm{FO}^k(\tau)$ with $f(\varphi) = r$, that is
separating for $(\sA,\sB)$.
\end{thm}

\begin{proof}\hfill

$(\Longrightarrow:)$ Suppose Spoiler closes $\cT$ even with optimal play from Duplicator. We use induction on the closed game tree $\cT$ to prove that there is a separating formula. If $\cT$ consists of a single node, then Spoiler closed the root node for his first move, and this corresponds to an atomic formula with the specified counter value by definition of the $\close$ move. Otherwise, consider a closed tree $\cT$, and consider Spoiler's first move on the root. Inductively, each of the subtrees corresponds to a separating formula. We consider each possible first move by Spoiler. Let the root node be $\Xr = \gnode{\sA}{\sB}{r}$.
\begin{itemize}
\item Suppose Spoiler plays $\swap$, creating $X = \gnode{\sB}{\sA}{r'}$. By induction, $(\sB, \sA)$ is separable by a formula $\psi$ such that $f(\psi) = r'$. Consider the formula $\varphi = \lnot\psi$, and observe that $\varphi$ is a separating formula for $(\sA, \sB)$. Furthermore, $f(\varphi) = h_\lnot(f(\psi)) = h_\lnot(r') = r$.
\item Suppose Spoiler plays $\splitleft$, creating $X_1 = \gnode{\sA_1}{\sB}{r_1}$ and $X_2 = \gnode{\sA_2}{\sB}{r_2}$. By induction, $(\sA_1, \sB)$ is separable by a formula $\varphi_1$ with $f(\varphi_1) = r_1$, and $(\sA_2, \sB)$ is separable by a formula $\varphi_2$ with $f(\varphi_2) = r_2$. Then, $\varphi = \varphi_1\lor\varphi_2$ is a separating formula for $(\sA, \sB)$, and $f(\varphi) = h_\lor(r_1, r_2) = r$.
\item If Spoiler plays $\splitright$, the analysis is similar.
\item Suppose Spoiler plays $\pebbleleft$ using the pebble color $x_i$, creating $X = \gnode{\sA'}{\sB'}{r'}$. By induction, $(\sA', \sB')$ has a separating formula $\psi$ with $f(\psi) = r'$. Let $\varphi = \exists x_i\psi$, and consider one particular structure-assignment pair $(\bA, \alpha_{\bA}) \in \cA$. We know this pair became a pair $(\bA, \alpha'_{\bA}) \in \sA'$, such that $(\bA, \alpha'_{\bA}) \models \psi$. Furthermore, the only difference between the two pairs is the mapping of the domain element $x_i$. Regardless of whether or not $x_i \in \mathsf{dom}(\alpha_{\bA})$, we have $(\bA, \alpha_{\bA}) \models \exists x_i\psi$, just by reverting the mapping of the element $x_i$ (or arbitrarily mapping it, if it was not in the domain). On the other hand, if some $(\bB, \alpha_\bB) \in \sB$ satisfies $(\bB, \alpha_\bB) \models \exists x_i\psi$, then Duplicator can respond to Spoiler's $\pebbleleft$ move by (re-)mapping the variable $x_i$ to a witness for $\psi$, and the resulting pair $(\bB, \alpha'_{\bB}) \in \sB'$ would satisfy $(\bB, \alpha'_{\bB}) \models \psi$. This is a contradiction, so every $(\bB, \alpha_\bB) \in \sB$ satisfies $(\bB, \alpha_\bB) \models \lnot\exists x_i\psi \equiv \lnot\varphi$. Therefore, $\varphi$ is a separating formula for $(\sA, \sB)$. Of course, $f(\varphi) = f(\exists x_i\psi) = h_\exists(f(\psi)) = h_\exists(r') = r$.
\item If Spoiler plays $\pebbleright$, the analysis is similar.
\end{itemize}

$(\Longleftarrow:)$ We use induction on the separating formula. Suppose there is a separating formula $\varphi \in \mathrm{FO}^k(\tau)$ for $(\sA, \sB)$ with $f(\varphi) = r$. We will show inductively that Spoiler has a winning strategy in the game $\SG{r}{k}{f}$ on $(\sA, \sB)$.

If $\varphi$ is atomic, Spoiler can just use $\varphi$ to close the root node
$\Xr=\gnode{\sA}{\sB}{r}$. Otherwise, inductively, we have the following cases:

\begin{itemize}
\item If $\varphi = \lnot\psi$, then Spoiler plays $\swap$, creating child $X_1=\gnode{\sB}{\sA}{r'}$ where $f(\psi)= r'$ and $r = h_\lnot(r')$. Note that the precondition is met, so Spoiler can play this move. Note that $\psi \in \mathrm{FO}^k(\tau)$ separates $(\sB,\sA)$, so by induction, Spoiler wins the remaining game.
\item If $\varphi = \psi\lor\gamma$, then every $(\bA,\alpha_\bA) \in \sA$ satisfies $(\bA, \alpha_\bA) \models \psi\lor\gamma$, whereas every $(\bB,\alpha_\bB) \in \sB$ satisfies $(\bB, \alpha_\bB) \models \lnot(\psi\lor\gamma) \equiv \lnot\psi \land \lnot\gamma$. Let $\sA_1$ be the subset of $\sA$ satisfying $\psi$, and let $\sA_2 = \sA - \sA_1$. Let $f(\psi) = r_1$ and $f(\gamma) = r_2$, so that $r = h_\lor(r_1, r_2)$. Now, Spoiler plays $\splitleft$ with nodes $\gnode{\sA_1}{\sB}{r_1}$ and $\gnode{\sA_2}{\sB}{r_2}$. Note that the precondition is met, so Spoiler can play this move. The formula $\psi \in \mathrm{FO}^k(\tau)$ separates $(\sA_1,\sB)$ and satisfies $f(\psi) = r_1$. Similarly, the formula $\gamma \in \mathrm{FO}^k(\tau)$ separates $(\sA_2,\sB)$ and satisfies $f(\gamma) = r_2$. By induction, Spoiler wins the remaining games.
\item If $\varphi = \psi\land\gamma$, Spoiler plays $\splitright$ analogously.
\item Suppose $\varphi = \exists x_i\psi$. If $f(\psi) = r'$, then note that $r = h_\exists(r')$. Note that the precondition is met, so Spoiler can play $\pebbleleft$. By assumption, every $(\bA, \alpha_\bA) \in \sA$ has $(\bA, \alpha_\bA) \models \exists x_i\psi$. So for some $a \in A$, $(\bA, \alpha'_\bA) \models \psi$, where $\alpha'_\bA$ is identical to $\alpha_\bA$ except that $x_i$ has been (re)mapped to the element $a$. Spoiler plays $\pebbleleft$, placing pebble $x_i$ on the witness $a$ on each $\bA$. Note that this creates exactly $(\bA, \alpha'_\bA)$. On the other hand, every $(\bB, \alpha_\bB) \in \sB$ satisfies $(\bB, \alpha_\bB) \models \lnot\exists x_i\psi \equiv \forall x_i\lnot\psi$, and so every $b \in B$ is a witness for $\lnot\psi$, and so even after making copies, Duplicator can never play on a witness for $\psi$ in any $\bB$. So for each pair $(\bB, \alpha'_\bB)$ created by Duplicator, we have $(\bB, \alpha'_\bB) \models \lnot\psi$. At the end, therefore, we have the new tree node $X' = \gnode{\sA'}{\sB'}{r'}$, where $\psi\in\mathrm{FO}^k(\tau)$ is a separating formula for $(\sA', \sB')$ with $f(\psi) = r'$, and so by induction, Spoiler wins the remaining game.
\item If $\varphi = \forall x_i\psi$, Spoiler plays $\pebbleright$ analogously.
\end{itemize}
This proves both directions of the theorem.
\end{proof}

We remark here that Theorem \ref{thm:syntacticchar} is applicable directly to several straightforward compositional syntactic measures, including the ones in Examples \ref{exa:quantcount}, \ref{exa:quantrank}, and \ref{exa:formulasize}. In particular, an examination of Example \ref{exa:quantcount} shows that the $(r, k)$-\qvt game is identical to $\SG{r}{k}{f_q}$, where $f_q$ is the quantifier-number measure; this proves Theorem \ref{thm:qvtchar} immediately.

We now complete this subsection with a result showing that a familiar tool for analyzing these games carries over to this general framework as well. The following proposition shows that Duplicator's oblivious strategy is optimal in the syntactic game for \emph{every} compositional syntactic measure.

\begin{prop}\label{prop:obliviousST}
Let $f$ be a compositional syntactic measure. If Duplicator has a winning strategy in the $\SG{r}{k}{f}$ on $(\sA, \sB)$, then the oblivious strategy is winning.
\end{prop}

\begin{proof}
It suffices to consider the case where Spoiler plays $\pebbleleft$ on some node $X = \gnode{\sA'}{\sB'}{r'}$ of $\cT$, creating the new node $X' = \gnode{\sA''}{\sB''}{r''}$. Note that $h_\exists(r'') = r'$. The $\pebbleright$ move will be analogous, and Duplicator has no agency elsewhere.

Suppose there is no separating formula $\varphi \in \mathrm{FO}^k(\tau)$ for $(\sA', \sB')$ with $f(\varphi) = r'$. It suffices to show that regardless of Spoiler's move, as long as Duplicator responds obliviously, there is no separating formula $\psi \in \mathrm{FO}^k(\tau)$ for $(\sA'',\sB'')$ with $f(\psi) = r''$.

Suppose WLOG Spoiler plays the pebble color $x_i$ on his move. Consider all pairs $(\bA, \alpha'_\bA) \in \sA''$, and take any arbitrary formula $\psi \in \mathrm{FO}^k(\tau)$ with $f(\psi) = r''$ that is true for \emph{all} these pairs. In other words, for every pair $(\bA, \alpha'_\bA) \in \sA''$, we have $(\bA, \alpha'_\bA) \models \psi$. An arbitrary pair $(\bA, \alpha'_\bA) \in \sA''$ arose from a pair $(\bA, \alpha_\bA) \in \sA'$, with Spoiler changing the assignment $\alpha_\bA$ by just (re)assigning the variable $x_i$. Furthermore, every pair $(\bA, \alpha_\bA) \in \sA'$ gave rise to a new pair in this way.

We claim that this means the formula $\sigma = \exists x_i\psi$ is true of all pairs $(\bA, \alpha_\bA) \in \sA'$. This follows by simply reversing the assignment of the variable $x_i$ in each of the pairs. It follows that for each pair $(\bA, \alpha_\bA) \in \sA'$, we have $(\bA, \alpha_\bA) \models \sigma$. Note that $f(\sigma) = f(\exists x_i\psi) = h_\exists(f(\psi)) = h_\exists(r'') = r'$. By our assumption, therefore, there must be some pair $(\bB, \alpha_\bB) \in \sB'$ satisfying $(\bB, \alpha_\bB) \models \sigma$. In other words, $(\bB, \alpha_\bB) \models \exists x_i\psi$.

Now, when Duplicator responds obliviously, she (re)assigns the variable $x_i$ in all possible ways in each of the pairs in $\sB'$. In particular, therefore, in one of the copies of the pair $(\bB, \alpha_\bB)$ defined above, she plays on a witness $b \in B$ for $\psi$, creating the pair $(\bB, \alpha'_\bB) \in \sB''$. But then, $(\bB, \alpha'_\bB) \models \psi$, and so $\psi$ cannot be a separating formula for $(\sA'', \sB'')$.

Since $\psi$ was arbitrary, we are done.
\end{proof}

Of course, Proposition \ref{prop:obliviousQVT} follows immediately as a corollary.

It is worth remarking that the proof of Proposition \ref{prop:obliviousST} also captures why it is critical for Duplicator to use her copying ability. Duplicator needs a pair $(\bB, \alpha_\bB) \in \sB'$ to be a witness for the formula $\psi$. However, the same pair can be a witnessing structure for two different such formulas. In particular, it could be the case that $(\bB, \alpha_\bB) \models \exists x_i\psi_1$ and $(\bB, \alpha_\bB) \models \exists x_i\psi_2$, but the elements in $B$ that witness $\psi_1$ and $\psi_2$ are \emph{different}. Making copies of the pair $(\bB, \alpha_\bB)$ in order to have two new assignments of $x_i$, enabling different elements to act as witnesses simultaneously in the same move, achieves this.

\subsection{Properties of some measures}

Having defined the syntactic game, it is natural to ask for the properties of specific syntactic measures that can enable us to apply existing tools of analyzing similar games, in order to have the results carry over to this realm.

Consider the quantifier rank syntactic measure, denoted by $f_r$ in Example \ref{exa:quantrank}. Consider the game $\SG{r}{k}{f_r}$, and observe that it is very similar to the \qvt game. The only difference is that in the $\splitleft$ and $\splitright$ moves on a tree node $X = \gnode{\cA'}{\cB'}{r'}$, Spoiler labels the new nodes in $\cT$ with two values whose maximum is $r'$, instead of two values that sum to $r'$.

The following proposition considers when we could look at games on each of the singleton pairs, rather than entire classes of $\tau$-structures.

\begin{prop}\label{singletonnicemeasure}
Consider the syntactic game for any compositional syntactic measure $f$ satisfying:
\begin{enumerate}[1.]
\item there are finitely many equivalence classes of the equivalence relation ``satisfy the same {\fo}-sentences with $f$-measure $r$ and $k$ variables'',
\item each such class is definable by a sentence with the parameters defining the class,
\item the set of sentences with these parameters is closed under disjunction.
\end{enumerate}
Then, for all nonempty sets $\cA$ and $\cB$ of $\tau$-structures, and for all $r, k \in \N$, the following statements are equivalent:
    \begin{enumerate}
        \item Spoiler wins the $\SG{r}{k}{f}$ on $(\cA, \cB)$.
        \item For every $\bA \in \cA$ and every $\bB \in \cB$, there is a separating $k$-variable sentence $\varphi$ for $(\{\bA\}, \{\bB\})$ with $f(\varphi) = r$.
    \end{enumerate}
\end{prop}

\begin{proof}\hfill

$(\Longrightarrow:)$ Take a separating $k$-variable sentence $\varphi$ for $(\cA,\cB)$ with $f(\varphi) = r$. This is separating for every pair $(\{\bA\}, \{\bB\})$ with $\bA \in \cA$ and $\bB \in \cB$.

$(\Longleftarrow:)$ By property 1, there are finitely many equivalence classes of the equivalence relation ``satisfy the same {\fo}-sentences with measure $r$ and $k$ variables'', and by property 2, each such class is definable by a {\fo}-sentence with these parameters. Let $\theta_1,\ldots,\theta_m$ be the list of all such  sentences arising from structures in $\cA$, and consider the sentence $\psi = \theta_1 \lor \ldots \lor \theta_m$. By property 3, $\psi$ is a $k$-variable formula with $f(\psi) = r$. We claim that $\psi$ is a separating sentence for $(\cA, \cB)$. Indeed, clearly, every structure in $\cA$ satisfies one of the $\theta_j$'s. We claim that no structure $\bB \in \cB$ does. Otherwise, take $\bB \in \cB$ satisfying $\theta_j$ for some $j$. Then $\bB$ satisfies the same sentences with measure $r$ and $k$ variables as some structure $\bA \in \cA$. But then $(\{\bA\},\{\bB\})$ is inseparable, which is a contradiction.
\end{proof}

Proposition \ref{singletonnicemeasure} immediately helps us relate the syntactic game on the quantifier rank measure $f_r$ to the other canonical technique we have for quantifier rank --- the \ef games.

\begin{cor}\label{singletonqd}
Let $f_r$ be the quantifier rank syntactic measure. For all nonempty sets $\cA$ and $\cB$ of $\tau$-structures, and for all $r, k \in \N$, the following statements are equivalent:
\begin{enumerate}
\item Spoiler wins the game $\SG{r}{k}{f_r}$ on $(\cA, \cB)$.
\item For every $\bA \in \cA$ and every $\bB \in \cB$, Spoiler wins the $\ef(r,k)$ game on $(\bA, \bB)$.
\end{enumerate}
\end{cor}

We next remark on the formula size measure, denoted by $f_s$ in Example \ref{exa:formulasize}. Consider the $\SG{r}{k}{f_s}$ game, and observe that it would be very similar to the \qvt game, with the following differences on a node $X = \gnode{\cA'}{\cB'}{r'}$:
\begin{itemize}
\item For a $\pebbleleft$, $\pebbleright$, or $\swap$ move, Spoiler would decrement the value of the counter $r'$ to $r' - 1$.
\item For a $\splitleft$ or $\splitright$ move on a node with counter value $r'$, Spoiler would choose $r'_1$ and $r'_2$ such that $r' = r'_1 + r'_2 + 1$.
\item Spoiler would only close nodes with counter value of $1$.
\end{itemize}

Note that this last game is very similar to the ones in \cite{adlerimmerman} and \cite{hv15}. In fact, as in \cite{hv15}, we could have changed the rules of the syntactic game and allowed Duplicator to choose one of the branches from a split move to continue along as well, without affecting the eventual winner. But then we would lose the property that a closed game tree is isomorphic to a parse tree of a separating sentence.

Finally, we remark briefly on the measure $f_t$ that is always zero. This is clearly a valid compositional syntactic measure. It should be clear that $\SG{0}{k}{f_t}$ is equivalent to the $k$-pebble game \cite{DBLP:journals/jsyml/Barwise77,DBLP:journals/jcss/Immerman82}. Formally, we claim that Spoiler wins the $k$-pebble game on $(\bA, \bB)$ if and only if he wins $\SG{0}{k}{f_t}$ on $(\{\bA\}, \{\bB\})$. Indeed, if Spoiler has a winning strategy on the $k$-pebble game, he can make the same move at the root node of $\cT$, and then once Duplicator responds, he can split the resulting node into singletons and win inductively on each of them. Conversely, if Duplicator wins the $k$-pebble game, she just plays the same moves on the syntactic game without making any copies, proving the result.
\section{Open Problems \& Future Directions}\label{sec:openproblems}

Our results suggest several open problems about \MS games and their
variants. To begin with, we saw that the hereditary \ms game with
repebbling (Game \ref{msgame2}) does not simultaneously capture the
number of quantifiers and the number of variables. What is the
fragment of {\fo} captured by this game? Considered from a slightly
different perspective, how is the ``hereditary'' property reflected by
{\fo} semantics?

We also investigated the variant of the \ms game in which Spoiler wins
without ever playing on top. In some cases, Spoiler wins the \ms game
on $(\cA,\cB)$ without playing on top on the left side (but may play
on top on the right), while in other cases, this is reversed. Is it
true that Spoiler wins the $r$-round \ms game on $(\cA,\cB)$ without
playing on top if and only if Spoiler wins both the $r$-round \ms game
on $(\cA,\cB)$ without playing on top on the left side, and the
$r$-round \ms game on $(\cA,\cB)$ without playing on top on the right
side?

We investigated compositional syntactic measures and their associated
games for first-order logic (introduced for the case of infinitary formulas by \cite{DBLP:journals/mlq/VaananenW13}).  These definitions allowed us to relate structural properties
of such complexity measures to properties of game strategies.  What
other relationships are possible?  Can we unify and better understand
logical inexpressibility results by studying syntactic measures?

Beyond these and other related problems, the main challenge is to use
syntactic games to prove lower bounds on any reasonable complexity
measure, first on unordered structures and then, hopefully, on ordered
structures --- the latter will be very difficult, because of
connections to longstanding open problems in computational complexity
(see Section \ref{sec:intro}). 
A few exceptions notwithstanding (e.g., \cite{DBLP:conf/focs/Schwentick94,DBLP:conf/lics/Ruhl99, DBLP:books/daglib/0009850}), there are relatively few inexpressibility results obtained using  EF games on ordered structures or on structures with built-in arithmetic predicates.  Can we go beyond these applications of the EF games? What about other syntactic measures on ordered structures, such as simultaneous bounds on number of quantifiers and number of variables?

\section*{Acknowledgments}
\noindent We would like to thank Ryan Williams (MIT) for many helpful discussions. Rik Sengupta acknowledges NSF grant CCF-1934846. We also would like to thank the anonymous referees for an incredibly careful reading of the paper, and for many helpful suggestions and edits.

\bibliographystyle{alphaurl}
\bibliography{references}

\end{document}